\documentclass[draftcls,onecolumn]{IEEEtran}
\usepackage{setspace}
\usepackage[margin=1in]{geometry}
\usepackage{color}
\usepackage{csquotes}
\usepackage{enumitem}
\usepackage{float}
\usepackage{multirow}

\floatstyle{ruled}
\newfloat{algorithm}{tbp}{loa}
\providecommand{\algorithmname}{Algorithm}
\floatname{algorithm}{\protect\algorithmname}
\usepackage{hyperref}
\usepackage{amsopn}

\makeatletter
\newcommand{\manuallabel}[2]{\def\@currentlabel{#2}\label{#1}}
\makeatother
\usepackage{amsmath}
\usepackage{amssymb}
\usepackage{amsthm}
\usepackage{bbm}
\usepackage{tikz}
\usepackage{mathrsfs}
\usepackage{pgfplots}
\pgfplotsset{compat=1.14}
\usetikzlibrary{arrows}
\usetikzlibrary{arrows.meta}

\newtheorem{theorem}{Theorem}
\newtheorem{lemma}[theorem]{Lemma}

\newtheorem{dfn}[theorem]{Definition}
\newtheorem{remark}[theorem]{Remark}
\newtheorem{cor}[theorem]{Corollary}

\newtheorem{assump}[theorem]{Assumption}

\newcommand{\calI}{\mathcal{I}}

\newcommand{\calP}{\mathcal{P}}

\newcommand{\db}{d_{\rm B}}

\newcommand{\p}{\mathbb{P}}
\newcommand{\RR}{\mathbb{R}}
\newcommand{\EE}{\mathbb{E}}

\renewcommand{\epsilon}{\varepsilon}

\title{Statistical Mean Estimation with \\ Coded Relayed Observations}

\author{Yan Hao Ling, Zhouhao Yang, and Jonathan Scarlett\thanks{Y.H.~Ling is with the  Department of Computer Science, School of Computing, National University of Singapore (NUS). 

Z.~Yang is with the Department of Applied Mathematics and Statistics at Johns Hopkins University.

J.~Scarlett is with the Department of Computer Science, Department of Mathematics, and Institute of Data  Science, National university of Singapore.

This research is supported by the Singapore National Research Foundation under its Global AI Visiting Professorship program.

Emails: \url{lingyh@nus.edu.sg}; \url{zyang145@jh.edu}; \url{scarlett@comp.nus.edu.sg}}}

\begin{document}
\maketitle

\begin{abstract}
    We consider a problem of statistical mean estimation in which the samples are not observed directly, but are instead observed by a relay (``teacher'') that transmits information through a memoryless channel to the decoder (``student''), who then produces the final estimate.  We consider the minimax estimation error in the large deviations regime, and establish achievable error exponents that are tight in broad regimes of the estimation accuracy and channel quality.  In contrast, two natural baseline methods are shown to yield strictly suboptimal error exponents.  We initially focus on Bernoulli sources and binary symmetric channels, and then generalize to sub-Gaussian and heavy-tailed settings along with arbitrary discrete memoryless channels.
\end{abstract}

\section{Introduction}

In this paper, we are interested in the fundamental statistical problem of mean estimation, in which $n$ samples $X_1,\dotsc,X_n$ are drawn independently from some unknown distribution $P_X$, and we are interested in estimating $\theta^* := \EE[X]$.  The related work on mean estimation is extensive, and is briefly discussed in Section \ref{sec:related}.

The main distinction from standard mean estimation in this paper is that the samples are not observed directly by the estimator, but are instead observed by an intermediate agent and transmitted through a noisy communication channel.  The setup is summarized in Figure \ref{fig:setup}, and is detailed as follows:
\begin{itemize}
    \item There are two agents, which we call the \emph{teacher} and \emph{student} following the terminology of \cite{jog2020teaching} (alternatively, they may be referred to as the \emph{relay} and \emph{decoder}).
    \item The teacher sequentially observes $X_1,\dotsc,X_n$.
    \item At each time $t \in \{1,\dotsc,n\}$, the teacher sends information to the student through a single use of a discrete memoryless channel (DMC); the input at time $t$ is denoted by $W_t$, the output by $Z_t$, and the channel by $P_{Z|W}$.  Importantly, $W_t$ is only allowed to depend on the previous samples $X_1,\dotsc,X_{t-1}$.
    \item Based on $Z_1,\dotsc,Z_n$, the student forms an estimate $\hat{\theta}_n$ of $\theta^*$.
\end{itemize}
We will initially focus on the case that $X \sim {\rm Ber}(\theta)$ and $P_{Z|W}$ is a binary symmetric channel (BSC) with parameter $p \in \big(0,\frac{1}{2}\big)$, as this already captures the main ideas and difficulties.  We will then turn to more general sources and channels in the later sections.

In general, mean estimation can come with several different goals, such as mean squared error guarantees and small/moderate/large deviations bounds.  In this paper, we focus on the large deviations regime; our reasoning/motivation for this is discussed in Section \ref{sec:large_dev}.  Specifically, for some constant $\epsilon > 0$ (not depending on $n$), we seek to minimize
\begin{equation}
    \mathbb{P}(|\hat{\theta}_n - \theta^*| > \epsilon)
\end{equation}
and analyze how quickly this quantity decays as a function of $n$.  We will generally assume that $\epsilon$ and $P_{Z|W}$ are known throughout the system, though in our protocols only the student will use this knowledge; the teacher will not, except ``indirectly'' in the sense of knowing a good error-correcting code for the channel $P_{Z|W}$.
  
We say that a teaching and learning strategy achieves an error exponent $E$ for a family of distributions $\calP_X$ if
\begin{equation}
    \inf_{P_X \in \calP_X} \limsup_{n \rightarrow \infty} -\frac1n \log \mathbb{P}(|\hat{\theta}_n - \theta^*| > \epsilon) \ge E, \label{eq:E_def}
\end{equation}
and we define $E^* = E^*(\calP_X,P_{Z|W},\epsilon)$ as the highest possible error exponent among all protocols. 
Our goal is to obtain upper and lower bounds on the optimal error exponent $E^*$, for various families $\calP_X$ of interest.  We highlight that our focus is on information-theoretic limits rather than computation or ``practical'' performance, though we will note in Section \ref{sec:computation} that our protocols have polynomial runtime.  We also note that our focus is the scalar-valued case in which $X$ takes values in $\mathbb{R}$, but we will discuss extensions to $\RR^d$ in Section \ref{sec:vector}.

\begin{figure*}[!t]
    \centering
    \begin{tikzpicture}
\draw (3.25,1.25) node {$X$};
\draw (5.75,1.25) node {$W$};
\draw (8.25,1.25) node {$Z$};
\draw (10.75,1.25) node {$\hat{\theta}_n$};
\draw[->] (3,1) -- (3.5,1);
\draw[thick] (3.5,0.5) -- (3.5,1.5) -- (5.5,1.5) -- (5.5,0.5) -- (3.5,0.5);
\draw[->] (5.5,1) -- (6,1);
\draw[thick] (6,0.5) -- (6,1.5) -- (8,1.5) -- (8,0.5) -- (6,0.5);
\node at (4.5, 1) {Teacher};
\node at (7, 1) {$P_{Z|W}$};
\draw[->] (8,1) -- (8.5,1);
\draw[thick] (8.5,0.5) -- (8.5,1.5) -- (10.5,1.5) -- (10.5,0.5) -- (8.5,0.5);
\node at (9.5, 1) {Student};
\draw[->] (10.5,1) -- (11,1);
\end{tikzpicture}
    \caption{Illustration of our problem setup; the quantity being estimated is $\theta^* = \EE[X]$.}
    \label{fig:setup}
\end{figure*}
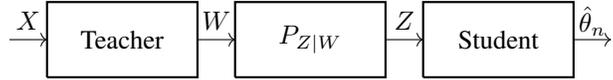

\subsection{Related work} \label{sec:related}

{\bf 1-bit and low-rate teaching/learning.} Our work is closely related to a recent line of works on teaching/learning a single bit of information \cite{jog2020teaching} and the essentially equivalent problem of error exponents for low-rate relaying over a tandem of channels \cite{onebit}.  The problem studied in \cite{jog2020teaching,onebit} resembles our problem specialized to Bernoulli+BSC, but the goal is to estimate a single bit observed by the teacher through another BSC, rather than estimating the Bernoulli distribution parameter.  Hence, a fundamental difference is the distinction between learning a discrete quantity (e.g., a single bit) vs.~learning a continuous quantity (e.g., a Bernoulli parameter).  We will find that some ideas such as block-structured teaching are common to both problems, but the details and analysis are largely distinct; see Section \ref{sec:baselines} for further discussion.

For the 1-bit teaching/learning problem, following the initial works in \cite{jog2020teaching,onebit}, the optimal error exponent was identified for BSCs in \cite{teachlearn}, and for general binary-input DMCs in \cite{ling2024optimal}.  Extensions to multi-bit and multi-hop settings were given in \cite{multibit,ling2023maxflow}.  It was noted in \cite{onebit} that the problem has connections to a ``many-hop'' (number of hops linear in the block length) problem called information velocity, which was subsequently studied in \cite{infovelocity,domanovitz2022information,domanovitz2024information}.

{\bf Mean estimation and large deviations theory.} Mean estimation is one of the most fundamental and widely-studied problems in statistics, and we do not attempt to summarize the literature in detail.  For distributions with light tails, a typical strategy is to estimate using the sample mean, and then apply concentration inequalities \cite{boucheron,vershynin} and/or large deviations analyses \cite{dembo2009large}.  For heavy-tailed distributions, the sample mean is typically a very poor estimate in the large deviations regime, and other estimators are adopted such as the trimmed mean or median-of-means \cite{lugosi2019mean}.  

Although the problem that we study is one of mean estimation, we will heavily rely on tools from binary hypothesis testing (e.g., see \cite{berlekampI,berlekampII,blahut2003hypothesis}), and our strategy will be to combine the results of several hypothesis tests to produce the final mean estimate.  In Remark \ref{rem:bernoulli}, we will discuss how this hypothesis testing approach can outperform approaches based on the sample mean, even in the setting of direct observations.

{\bf Communication-constrained statistical problems.} There have recently been a wide range of works studying estimation and other statistical problems under various kinds of communication constraints, e.g., see \cite{acharya2021distributed,duchi,ozgur,han2018geometric,salehkalaibar,gao2022review} and the references therein.  Certain classical works also fall in this category, e.g., \cite{ahlswede1986hypothesis}.  However, we are unaware of any such works that closely align with our problem setup; the works that we are aware of have significantly more differences to ours (compared to the ``closer'' 1-bit teaching/learning problems outlined above) and adopt substantially different methods.  For example, notable differences include considering other recovery criteria such as mean squared error \cite{duchi,ozgur}, considering noiseless bits instead of noisy channels \cite{acharya2021distributed,duchi,ozgur,ahlswede1986hypothesis}, considering distinct problems such as distribution learning and testing \cite{acharya2021distributed,salehkalaibar,ahlswede1986hypothesis,acharya2021interactive}, and considering problems where all samples are observed before coding \cite{ahlswede1986hypothesis}.

% these various works are generally incomparable to ours due to the different problem settings, but we highlight some more closely related works as follows. 
% {\bf \color{blue} [TODO: Maybe highlight benefits of interactivity \cite{acharya2021interactive} or multiple samples per user \cite{acharya2021distributed}.  And should we discuss connections of our setting to the distributed setting?]}
% {\color{blue}
% \begin{itemize}
%     \item \cite{acharya2021interactive}: Distribution learning and identity testing (not mean estimation).  Focuses more on $\epsilon$ dependence for constant error probability.  Focuses more on scaling laws in terms of alphabet size and $\epsilon$.  Interactivity is different (can depend on past outputs but not past inputs).
%     \item \cite{acharya2021distributed}: Distribution learning.  Focuses more on average minimax risk for TV distance.  Considers noiseless 1-bit quantization.
%     \item \cite{duchi}: Consider MSE; considers noiseless bits.
%    \item \cite{ahlswede1986hypothesis}: Non-causal hypothesis testing
% \end{itemize}
% }

\subsection{Discussion on large deviations and other regimes} \label{sec:large_dev}

In principle, our problem setup could be considered alongside several recovery criteria, including mean squared error, constant error probability, moderate deviations, or large deviations bounds.  We focus on the latter for two main reasons:
\begin{enumerate}
    \item In regimes other than large deviations, very simple strategies can be near-optimal in terms of the leading asymptotic terms (see below), and thus understanding the benefit of more complex protocols would require introducing higher-order asymptotic terms or moving to a non-asymptotic analysis.
    \item The large deviations regime closely aligns with the recent line of works on 1-bit teaching/learning and low-rate relaying \cite{onebit,jog2020teaching,teachlearn,multibit}, which itself has interesting connections to other problems such as information velocity \cite{onebit,infovelocity}.
\end{enumerate}
To elaborate on the first of these points, suppose that we consider a regime in which the deviation probability decays sub-exponentially in $n$, say with $\epsilon_n \to 0$ in a manner such that we can expect $\mathbb{P}(|\hat{\theta}_n - \theta^*| > \epsilon_n) \le e^{-\Theta(n^{0.99})}$.  Then fix $\delta \in (0,1)$ and consider a protocol in which the teacher estimates $\theta^*$ using the first $(1-\delta)n$ samples, and transmits a slightly quantized version of that estimate using the last $\delta n$ channel uses.  Since $\delta n$ is linear in $n$, we can use $n^C$ quantization levels (for any fixed $C > 0$) while still maintaining exponential decay in the error probability of transmitting the quantized estimate.  Thus, we attain essentially identical performance to directly estimating $\theta^*$ from $(1-\delta)n$ samples, which amounts to ``losing'' an arbitrarily small fraction of the samples.  As we will see in Lemma \ref{lem:est_forward} below and our numerical evaluations in Section \ref{sec:numerical}, the analogous loss of such an approach can become much more significant in the large deviations regime.

\subsection{Discussion of protocols} \label{sec:baselines}

In this subsection, we discuss a number of baselines and simple protocols, as well as briefly summarizing our own approach.  The corresponding exponents will be compared numerically in Section \ref{sec:numerical}.

Here and throughout the paper, it will be useful to write the results in terms of two error exponents associated with the source and the channel individually, defined as follows.

\begin{dfn}
    For a given class $\calP_X$ of source distributions and an accuracy parameter $\epsilon > 0$, we define the \emph{source exponent} (or \emph{direct observation exponent}) $E^{\rm src}_{\epsilon}(\calP_X)$ as the highest $E$ such that the error exponent $E$ is achieved (in the sense of \eqref{eq:E_def}) by some $\hat{\theta}_n$ formed \emph{directly} from the $n$ samples $X_1,\dotsc,X_n$.
\end{dfn}

\begin{dfn}
    For a discrete memoryless channel $P_{Z|W}$ and a number of messages $M$, we define the \emph{channel exponent} $E^{\rm chan}_M(P_{Z|W})$ as the highest achievable error exponent (in the standard sense \cite{berlekampI}; see Section \ref{sec:channel_coding}) for sending $M$ messages in $n$ channel uses.  Moreover, we define the zero-rate error exponent $E^{\rm chan}(P_{Z|W},0) = \lim_{M \to \infty} E^{\rm chan}_M(P_{Z|W})$ (see Lemma \ref{lem:zero_rate} below for a justification of this equality).
\end{dfn}

When the class of distributions and/or the channel are clear from the context, we adopt the shorthands $E^{\rm src}_{\epsilon}$, $E^{\rm chan}_M$, and $E^{\rm chan}(0)$.  We briefly note some expressions for these quantities in special cases of interest (with $D(a\|b) = a\log\frac{a}{b} + (1-a)\log\frac{1-a}{1-b}$ being the binary KL divergence):
\begin{itemize}
    \item For Bernoulli sources, we have  $E^{\rm src}_{\epsilon} = D\big( \frac{1}{2} \|  \frac{1}{2} + \epsilon \big)$ for $\epsilon \in \big(0,\frac{1}{2}\big)$ (see Section \ref{sec:bernoulli}).
    \item For Gaussian sources with mean in $[0,1]$ and variance $\sigma^2$, we have $E^{\rm src}_{\epsilon} = \frac{\epsilon^2}{2\sigma^2}$ (see Section \ref{sec:subgaussian}).
    \item For the BSC, we have $E^{\rm chan}_M = c_M \cdot (1/2) D(1/2 \| p)$ with $c_M \approx \frac{M}{M-1} \in [1/2,1]$, and $E^{\rm chan}(0) = (1/2) D(1/2 \| p)$ (see Lemma \ref{lem:cM} below, which also makes the ``$\approx$'' statement more precise).
\end{itemize}
We now proceed to outline various protocols and their associated error exponents.

\subsubsection{Non-causal setting} Recall that an important feature of our setting is that the $i$-th transmitted symbol $W_i$ can only depend on the previous samples $X_1,\dotsc,X_{i-1}$.  Nevertheless, it is useful to consider a hypothetical \emph{non-causal setting} where the teacher \emph{first} receives $X_1,\dotsc,X_n$ and \emph{then} transmits the entire sequence $W_1,\dotsc,W_n$.  In this setup, we have the following.

\begin{lemma} \label{lem:noncausal}
    In the non-causal setting, for any distribution class $\calP_X$ whose set of possible means is $\Theta^* = [0,1]$,\footnote{This assumption is trivial for the Bernoulli class, and we will discuss in Section \ref{sec:discussion} how it can be relaxed in other cases.} and any $\epsilon \in \big(0,\frac{1}{2}\big)$, we have the following:
    \begin{itemize}
        \item (Converse) Any non-causal protocol achieves error exponent at most $\min\{E^{\rm src}_{\epsilon},E^{\rm chan}_{\lceil 1/(2\epsilon) \rceil} \}$.
        \item (Achievability) For any $\delta \in (0,1)$, there exists a non-causal protocol achieving error exponent at least $\min\{E^{\rm src}_{(1-\delta)\epsilon},E^{\rm chan}_{\lceil 1/(2\delta\epsilon)\rceil} \}$.
    \end{itemize}
\end{lemma}
\begin{proof}
    See Appendix \ref{sec:noncausal}.
\end{proof}

Observe that in the limit as $\delta \to 0$, the achievable exponent approaches 
\begin{equation}
    \min\{ E^{\rm src}_{\epsilon},E^{\rm chan}(0) \}. \label{eq:delta0_exp}
\end{equation}
Thus, whenever the first term achieves this minimum, this protocol achieves the \emph{optimal direct-observation exponent} $E^{\rm src}_{\epsilon}$, i.e., the exponent that would be achieved if the student had direct access to $X_1,\dotsc,X_n$. 

\subsubsection{One-shot estimate-and-forward} Returning to the causal setting, a straightforward protocol is to perform the following for some parameters $\lambda,\delta \in (0,1)$:
\begin{itemize}
    \item Use the first $(1-\lambda)n$ samples to estimate $\theta^*$ to within accuracy $(1-\delta)n$;
    \item Use the last $\lambda n$ samples to transmit the estimate to within accuracy $\delta n$.
\end{itemize}
This simple protocol (which is described more formally in the proof of Lemma \ref{lem:est_forward}) can be analyzed in a similar manner to the non-causal setting, giving the following.

\begin{lemma} \label{lem:est_forward}
    For any distribution class $\calP_X$ whose set of possible means is $\Theta^* = [0,1]$, the one-shot estimate-and-forward protocol described above with parameters $(\lambda,\delta)$ attains an error exponent of $\min\{ (1-\lambda)E^{\rm src}_{(1-\delta)\epsilon}, \lambda  E^{\rm chan}_{\lceil 1/(2\delta\epsilon)\rceil} \}$. 
\end{lemma}
\begin{proof}
    See Appendix \ref{sec:oneshot}.
\end{proof}

We see that this approach suffers from shortening the ``effective block length'' in each term, thus multiplying the first term by $1-\lambda$ and the second term by $\lambda$.  In particular, this means that the exponent is strictly worse than the direct-observation exponent whenever $E^{\rm src}_{\epsilon} > 0$ and $E^{\rm chan}(0) < \infty$.  We note that a straightforward calculation reveals that for fixed $\delta$, the best choice in Lemma \ref{lem:est_forward} is $\lambda = \frac{ E^{\rm src}_{\lceil 1/(2\delta\epsilon)\rceil} }{E^{\rm src}_{(1-\delta)\epsilon} + E^{\rm chan}_{\lceil 1/(2\delta\epsilon)\rceil}}$, which equates the two terms in the $\min\{\cdot,\cdot\}$.

\subsubsection{Simple forwarding} In the case that the random variable $X$ has the same support as the channel input $W$, another natural strategy is \emph{simple forwarding}, in which the previous sample observed is transmitted directly, i.e., $W_i = X_{i-1}$.  In certain cases, the decoder can then estimate $\theta$ directly from the received sequence $Z_1,\dotsc,Z_n$; for concreteness, we demonstrate this idea in the specific case of a Bernoulli source and a BSC.

\begin{lemma} \label{lem:simple_forward}
    Consider the case of a Bernoulli source and a BSC with parameter $p \in \big(0,\frac{1}{2}\big)$, and fix $\epsilon \in \big(0,\frac{1}{2}\big)$.  When the teacher performs simple forwarding, there exists a design of the student such that an error exponent of $E^{\rm src}_{(1-2p)\epsilon}$ is achieved. 
\end{lemma}

A limitation of this protocol is that the randomness of the source and the noise from the channel get ``mixed together'', leading to an end-to-end distribution that is ``noisier'' than either of the two separately.  Since $E^{\rm src}_{\epsilon} = D\big( \frac12 \| \frac12 + \epsilon \big)$ is strictly increasing in $\epsilon$, we see that the exponent in Lemma \ref{lem:simple_forward} is is strictly worse than the direct-observation exponent $E^{\rm src}_{\epsilon}$ for all $p \in \big(0,\frac12\big)$.

% {\bf \color{red} [TODO: Should we discuss ``cumulative teaching''?  e.g., keep a running estimate, keep track of the number of 0s and 1s sent so far, and choose the next bit to move the empirical fraction closer to the current estimate.  Presumably this would have similar limitations to the 1-bit setup.]}

\subsubsection{Existing 1-bit teaching and learning protocol} In \cite{jog2020teaching,onebit,teachlearn}, the closely related problem of 1-bit 2-hop relaying was considered.  In their setup, there is only 1 bit (rather than a continuous parameter) to be estimated, but the teacher only receives observations of that bit passed through a BSC with a known crossover probability.

An asymptotically optimal strategy, proposed in \cite{teachlearn}, is to have the teacher read in blocks of size $k \ll n$, and transmit sequences of the form $1\dotsc10\dotsc0$ in size-$k$ blocks, where the number of 1s sent is chosen to be a carefully-designed function of the number of 1s received in the previous block.  One could conceivably consider a similar approach in our setup.  We do not make any claims on the degree of (sub)optimality of such a protocol, but we found it relatively difficult to analyze, and it was unclear to us what function should be used to decide the number of 1s sent.

\subsubsection{Our approach} We still use a block-structured strategy in a similar spirit as \cite{teachlearn} (see Section \ref{sec:bernoulli} for details), but within each block, we simply have the teacher send information about the sum or average of its received symbols using a \emph{general codebook} (whose codewords need not be of the form $1\dotsc10\dotsc0$).  We then consider a decoder that performs binary hypothesis tests on a number of $(\theta,\theta')$ pairs, constructs a set $S$ of $\theta$ values that are favored over all other values within distance $2\epsilon$, and outputs $\hat{\theta}_n$ as the midpoint of $S$.  See Section \ref{sec:bern_protocol} for the details.

\subsection{Summary of our results}

We now proceed to summarize our main results:
\begin{itemize}
    \item In Section \ref{sec:bernoulli}, we study the case of a Bernoulli source and a binary symmetric channel.  We provide a protocol that attains an error exponent of $\min\{ E^{\rm src}_{\epsilon}, E^{\rm chan}(0) \}$, and we show (via Lemma \ref{lem:noncausal}) that no protocol can attain an exponent better than $\min\{ E^{\rm src}_{\epsilon}, (1+2\epsilon+O(\epsilon^2)) E^{\rm chan}(0) \}$, where an exact expression is also available for the $\epsilon$-dependent factor, in particular equaling $\frac{1}{1-2\epsilon}$ under ``favorable rounding''.  Thus, we attain matching achievability and converse whenever the former $\min\{\cdot,\cdot\}$ is attained by the first term, and more generally a multiplicative gap of at most $1+2\epsilon+O(\epsilon^2)$.  In addition, the converse holds even for non-causal protocols.
    \item In Section \ref{sec:subgaussian}, we consider general sub-Gaussian sources and general discrete memoryless channels.  We provide a protocol that attains an error exponent of $\min\{ E^{\rm src}_{\epsilon}, E^{\rm chan}(0)\}$, and we know from Lemma \ref{lem:noncausal} that no protocol can attain an exponent better than $\min\{E^{\rm src}_{\epsilon},E^{\rm chan}_{\lceil 1/(2\epsilon) \rceil} \}$, where $E^{\rm chan}_{\lceil 1/(2\epsilon) \rceil} \ge E^{\rm chan}(0)$ but it holds that $E^{\rm chan}_{\lceil 1/(2\epsilon) \rceil} \to E^{\rm chan}(0)$ as $\epsilon \to 0$.  Thus, we attain matching achievability and converse whenever the former $\min\{\cdot,\cdot\}$ is attained by the first term, and even in the second term the multiplicative gap approaches 1 as $\epsilon$ decreases.  In addition, the converse holds even for non-causal protocols.
    \item In Section \ref{sec:discussion}, we extend the results of Section \ref{sec:subgaussian} beyond the sub-Gaussian case, in particular allowing general heavy-tailed distributions with a finite variance.  We also provide partial results for vector-valued sources, and we discuss the (polynomial-time) computational complexity of our protocols.
\end{itemize}

\section{Preliminaries} \label{sec:prelim}

In this section, we introduce some useful mathematical tools and results that will be used throughout the paper.

\subsection{Distance measures and distance-based codes}

For any two discrete distributions $P_1,P_2$ defined on the same set, we denote the Bhattacharyya coefficient by 
\begin{equation}
    \rho(P_1, P_2) = \sum_x \sqrt{P_1(x) P_2(x)},
\end{equation}
and the Bhattacharyya distance by 
\begin{equation}
    \db(P_1, P_2) = -\log \rho(P_1, P_2).
\end{equation}
For $\theta_1, \theta_2 \in [0,1]$, we will also adopt the shorthand
\begin{equation}
    \rho(\theta_1, \theta_2) = \rho({\rm Ber}(\theta_1), {\rm Ber}(\theta_2))
\end{equation}
to represent the Bhattacharyya coefficient between the corresponding Bernoulli distributions, and similarly for $\db(\theta_1,\theta_2)$ and $D(\theta_1 \| \theta_2)$.  The following simple identity will be useful.

\begin{lemma} \label{lem:eps_equiv}
    For any $\epsilon \in (0,1/2)$, it holds that 
    \begin{equation}
        \db(1/2-\epsilon,1/2+\epsilon) = D(1/2\|1/2+\epsilon) = D(1/2\|1/2-\epsilon) = -\frac{1}{2}\log(1-4\epsilon^2).
    \end{equation}
\end{lemma}
\begin{proof}
    For $\db$, a direct calculation gives $\rho(1/2-\epsilon,1/2+\epsilon) = 2\sqrt{(1/2-\epsilon)(1/2+\epsilon)} = 2\sqrt{1/4 - \epsilon^2}$, and taking the negative log and simplifying gives $-\frac{1}{2}\log(1-4\epsilon^2)$.  For the relative entropy, a direct calculation gives $D(1/2\|1/2+\epsilon) = \frac{1}{2}\log\frac{1/2}{1/2+\epsilon} + \frac{1}{2}\log\frac{1/2}{1/2-\epsilon}$, and we again get $-\frac{1}{2}\log(1-4\epsilon^2)$ by simplifying; note also that $D(1/2\|p) = D(1/2\|1-p)$.
\end{proof}

Consider a discrete memoryless channel $P_{Z|W}$, with input $W$ and output $Z$. Letting $\vec{W}, \vec{W}'$ be fixed strings of the same length over $W$, we adopt the shorthand
\begin{equation}
    \rho(\vec{W}, \vec{W}', P_{Z|W}) = \rho\big( \p(\cdot | \vec{W}), \p(\cdot | \vec{W}') \big)
\end{equation}
with $\p(\cdot | \vec{W})$ being the distribution of the output string $\vec{Z}$ when $\vec{W}$ is sent over the memoryless channel $P_{Z|W}$, and similarly for $\p(\cdot | \vec{W}')$.  Similarly, we write $\db(\vec{W}, \vec{W}', P_{Z|W}) = -\log \rho(\vec{W}, \vec{W}', P_{Z|W})$.
% We will also omit the channel $P_{Z|W}$ when it is clear from the context.

The following tensorization property of $\db$ is well known (e.g., see \cite{berlekampI}).

\begin{lemma} \label{lem:tens_db}
    For any two distributions $P_1,P_2$ of the form $P_1(\vec{x}) = \prod_{i=1}^k P_{1,i}(x_i)$ and $P_2(\vec{x}) = \prod_{i=1}^k P_{2,i}(x_i)$, it holds that 
    \begin{equation}
        \rho(P_1,P_2) = \prod_{i=1}^k \rho(P_{1,i},P_{2,i}), \quad \db(P_1,P_2) = \sum_{i=1}^k \db(P_{1,i},P_{2,i}).
    \end{equation}
    Consequently, given two codewords $\vec{W} = (W_1,\dotsc,W_k)$ and $\vec{W}' = (W'_1,\dotsc,W'_k)$ transmitted over a DMC $P_{Z|W}$, we have $\rho(\vec{W}, \vec{W}', P_{Z|W}) = \prod_{i=1}^k \rho(W_i, W'_i, P_{Z|W})$ and $\db(\vec{W}, \vec{W}', P_{Z|W}) = \sum_{i=1}^k \db(W_i, W'_i, P_{Z|W})$.
\end{lemma}

The following lemma states an achievable minimum distance for low-rate codes; this is a standard result, but we provide a short proof for completeness.

\begin{lemma} \label{lem:bsc_hamming_distance}
    For any constant $C > 0$, there exists a length-$k$ codebook with $k^C$ binary codewords and minimum Hamming distance $k/2 - o(k)$ as $k \to \infty$.
\end{lemma}
\begin{proof}
    Consider random coding, where each codeword is chosen uniformly at random from $\{0,1\}^k$. The Hamming distance between a pair of strings follows a Binomial$\big(k,\frac12\big)$ distribution. The probability of two specific codewords having Hamming distance at most $(\frac12-c)k$ is exponentially small in $c$. Since there are polynomially many codewords, a union bound over all pairs establishes the existence of a codebook with minimum distance $(\frac12-c)k$ when $k$ is large enough.  Since $c$ is arbitrarily small, the result follows.
\end{proof}

\subsection{Channel coding error exponents} \label{sec:channel_coding}

This subsection concerns error exponents for communication over discrete memoryless channels, not necessarily relating to mean estimation.  
For a fixed discrete memoryless channel $P_{Z|W}$ and integers $n$ and $M$, let $P_e^*(n,M)$ denote the smallest possible error probability for sending one of $M$ messages via $n$ uses of $P_{Z|W}$.  (For the purposes of error exponents, it is inconsequential whether this error probability is for the worst-case message or a uniformly random message.) 
For a given rate $R>0$, let
\begin{equation}
    E^{\rm chan}(R) = \lim_{n\rightarrow \infty} -\frac1n \log P_e^*(n, \exp(nR))
\end{equation}
be the associated optimal error exponent, and define $E^{\rm chan}(0)$ by continuity, i.e.
\begin{equation}
    E^{\rm chan}(0) = \lim_{R\rightarrow 0^+} E^{\rm chan}(R).
\end{equation}
Similarly, let
\begin{equation}
    E^{\rm chan}_M =  \lim_{n\rightarrow \infty} -\frac1n \log P_e^*(n, M)
\end{equation}
be the optimal error exponent for sending a constant number $M$ of messages.  Then, we have the following.

\begin{lemma} \label{lem:zero_rate}
    \emph{(\cite[Theorem 4]{berlekampI})}
    
    \begin{enumerate}
    \item For any discrete memoryless channel $P_{Z|W}$, we have 
    \begin{equation}
        E^{\rm chan}(0) = \max_{P_W} \sum_w \sum_{w'} P_W(w) P_W(w')\ \db(w,w',P_{Z|W}), \label{eq:e0_q}
    \end{equation}
    where $\max_{P_W}$ is taken over all probability distributions over the inputs of the channel
    \item It holds that
    \begin{equation}
        \lim_{M \rightarrow \infty} E^{\rm chan}_M = E^{\rm chan}(0).
        \label{eq:em_e0}
    \end{equation}    
    \end{enumerate}
\end{lemma}

For comparing various bounds, it will be useful to write $E^{\rm chan}_M$ as a fraction of $E^{\rm chan}(0)$, and for this purpose we introduce the following definition and lemma.

\begin{dfn} \label{def:cM}
    For any DMC $P_{Z|W}$, we define $c_M = c_M(P_{Z|W})$ as the value such that $E^{\rm chan}_M = c_M E^{\rm chan}(0)$.
\end{dfn}

\begin{lemma} \label{lem:cM}
    ({\em \cite[Thm.~3 and Cor.~3.2]{berlekampII} and \cite[Eq.~(1.19)]{berlekampI}})
    The quantity $c_M$ satisfies the following:
    \begin{itemize}
        \item[(i)] For the BSC, we have
            \begin{equation}
             c_M = 
             \begin{cases}
                 \frac{M}{M-1} & M \text{ is even} \\
                 \frac{4 \cdot \lfloor M/2 \rfloor \cdot \lceil M/2 \rceil}{M(M-1)} & M \text{ is odd}.
             \end{cases} \label{eq:cM_bsc}
        \end{equation}
        \item[(ii)] For any DMC, we have the following as $M \to \infty$:
        \begin{equation}
            c_M \ge \frac{M}{M-1} - O\Big( \frac{1}{M^2} \Big), \quad c_M \le 1 + O\Big( \frac{1}{\sqrt{\log \log M}} \Big),
        \end{equation}
        and hence $\lim_{M \to \infty} c_M = 1$.
    \end{itemize}
\end{lemma}

The following corollary for low rates will also be useful; this serves as a natural counterpart to Lemma \ref{lem:bsc_hamming_distance}.

\begin{cor}
    For any DMC $P_{Z|W}$ and any $C > 0$, there exists a length-$k$ codebook over $P_{Z|W}$ containing $k^C$ codewords such that $\db(\vec{W},\vec{W}
    ',P_{Z|W}) \geq (k - o(k))E^{\rm chan}(0)$ for any pair of distinct codewords $\vec{W}$ and $\vec{W}'$.
    \label{cor:db_dmc}
\end{cor}

\begin{proof} 
    The proof is analogous to that of Lemma \ref{lem:bsc_hamming_distance} but with $\db$ replacing the Hamming distance; see Appendix \ref{sec:pf_db_dmc} for the details.
\end{proof}

\begin{remark}
    We allow both the cases $E^{\rm chan}(0) < \infty$ and $E^{\rm chan}(0) = \infty$, but we note that in the latter case the optimal error exponent for mean estimation can be attained by simply using one-shot estimate-and-forward (Lemma \ref{lem:est_forward}) and taking $\lambda \to 0$ and $\delta \to 0$.
\end{remark}

\section{Bernoulli Source and Binary Symmetric Channel} \label{sec:bernoulli}

In this section, we prove our first main result, stated as follows.

\begin{theorem} \label{thm:bernoulli}
    {\em (Bernoulli+BSC Achievability)}
    Under a Bernoulli source, a BSC with parameter $p \in \big(0,\frac{1}{2}\big)$, and an accuracy parameter $\epsilon \in \big(0,\frac{1}{2}\big)$, the following error exponent is achievable:
    \begin{equation}
        E^* \geq \min\Bigg(D\bigg(\frac12 \Big\| \frac12 + \epsilon\bigg), \frac12 D\big(1/2 \big\| p\big)\Bigg). \label{eq:bern_ach}
    \end{equation}
\end{theorem}

\begin{remark} \label{rem:bernoulli}
    We make the following remarks on this result:
    \begin{enumerate}
        \item The exponent $D\big(\frac12 \big\| \frac12 + \epsilon\big)$ can in fact exceed what would be achieved by the sample mean estimator (in the setting of direct observations), which is $\min\big( D(\theta^* - \epsilon \| \theta^*), D(\theta^* + \epsilon \| \theta^*) \big)$ by the method of types or Sanov's theorem (e.g., setting $\theta^* = \frac{1}{2}$ and $\epsilon = 0.1$, the exponents are roughly 0.0204 and 0.0201 respectively).  Hence, the sample mean can be suboptimal when the recovery criterion is the error exponent for a given $\epsilon$.  On the other hand, the sample mean does not require knowledge of $\epsilon$, whereas our protocol does use such knowledge.
        \item We recall that $E^*$ is defined with respect to the worst-case choice of $\theta^* \in [0,1]$, so our result is minimax in nature.  Our proof will also reveal an instance-dependent achievable exponent (i.e., dependent on $\theta^*$; see \eqref{eq:e_theta} below), but for our converse in Corollary \ref{cor:bern_conv} below, we focus on minimax guarantees.
    \end{enumerate}    
\end{remark}

To understand the tightness of Theorem \ref{thm:bernoulli}, we state the following corollary of Lemma \ref{lem:noncausal}.

\begin{cor} \label{cor:bern_conv}
    {\em (Bernoulli+BSC Converse)}
    Under a Bernoulli source, a BSC with parameter $p \in \big(0,\frac{1}{2}\big)$, and an accuracy parameter $\epsilon \in \big(0,\frac{1}{2}\big)$, it holds (even in the non-causal setting) that
    \begin{equation}
        E^* \le \min\Bigg(D\bigg(\frac12 \Big\| \frac12 + \epsilon\bigg), c_{\lceil 1/(2\epsilon) \rceil} \cdot \frac{1}{2} D\big(1/2 \big\| p\big)\Bigg), \label{eq:bern_conv}
    \end{equation}
    where $c_{\lceil 1/(2\epsilon) \rceil} = 1+2\epsilon+O(\epsilon^2)$ as $\epsilon \to 0$, and the exact expression is given in \eqref{eq:cM_bsc}.
\end{cor}

Thus, the upper and lower bounds match whenever the $\min\{\cdot,\cdot\}$ in \eqref{eq:bern_ach} is achieved by the first term.  Moreover, the second terms match to within a multiplicative factor of $1+2\epsilon+O(\epsilon^2)$.  In certain scenarios the exact (non-asymptotic) multiplicative factor has a simple expression; in particular, if $\frac{1}{2\epsilon}$ is an even integer, then by \eqref{eq:cM_bsc} the factor is exactly $\frac{1}{1-2\epsilon}$ (i.e., $\frac{M}{M-1}$ with $M = \frac{1}{2\epsilon}$).

We will prove Theorem \ref{thm:bernoulli} in Sections \ref{sec:bern_protocol} and \ref{sec:pf_bern}, and we will prove Corollary \ref{cor:bern_conv} in Section \ref{sec:pf_bern_conv}.  

\subsection{Numerical comparisons} \label{sec:numerical}

We compare our Bernoulli+BSC exponents with the baselines from Section \ref{sec:baselines} in two ways:
\begin{itemize}
    \item In Figure \ref{fig:numerical1}, we fix the BSC parameter $p = 0.1$ and vary $\epsilon \in \big(0,\frac12\big)$;
    \item In Figure \ref{fig:numerical2}, we fix the accuracy parameter $\epsilon = 0.1$ and vary $p \in \big(0,\frac12\big)$.
\end{itemize}

We see that the the achievable error exponents for simple forwarding and one-shot estimate-and-forward strategies are mostly highly suboptimal, though simple forwarding becomes near-optimal for $\epsilon$ near $\frac{1}{2}$ (low accuracy)
%\footnote{In Figure \ref{fig:numerical1} simple forwarding may also appear to be near-optimal for $\epsilon$ near zero, but in fact a Taylor expansion gives $E_{\epsilon}^{\rm src} = D\big(\frac{1}{2} \| \frac{1}{2}+\epsilon\big) = 2\epsilon^2 + o(1)$, which means the exponent $E_{\epsilon(1-2p)}^{\rm src}$ of simple forwarding is suboptimal by a factor of roughly $(1-2p)^2$.} 
and one-shot estimate-and-forward becomes near-optimal for $\epsilon$ near $0$ (high accuracy) and for $p$ near $\frac12$ (high noise), i.e., when the source exponent is small or the channel exponent is small.

We see in Figure \ref{fig:numerical1} that our protocol has an \emph{optimal error exponent} that \emph{matches the direct-observation exponent $E_{\epsilon}^{\rm src}$} for a wide range of $\epsilon$ values of the form $[0,\epsilon_0]$; while not visible in this figure, the value of $\epsilon_0$ turns out to increase (resp., decrease) as $p$ decreases (resp., increases).  Moreover, we see in Figure \ref{fig:numerical2} that at least in this example with $\epsilon = 0.1$, our protocol achieves the optimal exponent for most values of $p$, and the gap to the converse is still fairly small for the remaining values of $p$.

We expect that improved protocols (both causal and non-causal) can be devised for $\epsilon$ close to $\frac12$, e.g., by specifically only using $M=2$ codewords of the form $0\dotsc0$ and $1\dotsc1$.  However, since we view this regime as being of less interest compared to small-to-moderate $\epsilon$, we do not pursue this point further.

    \begin{figure}
        \begin{centering}
            \includegraphics[width=0.4\columnwidth]{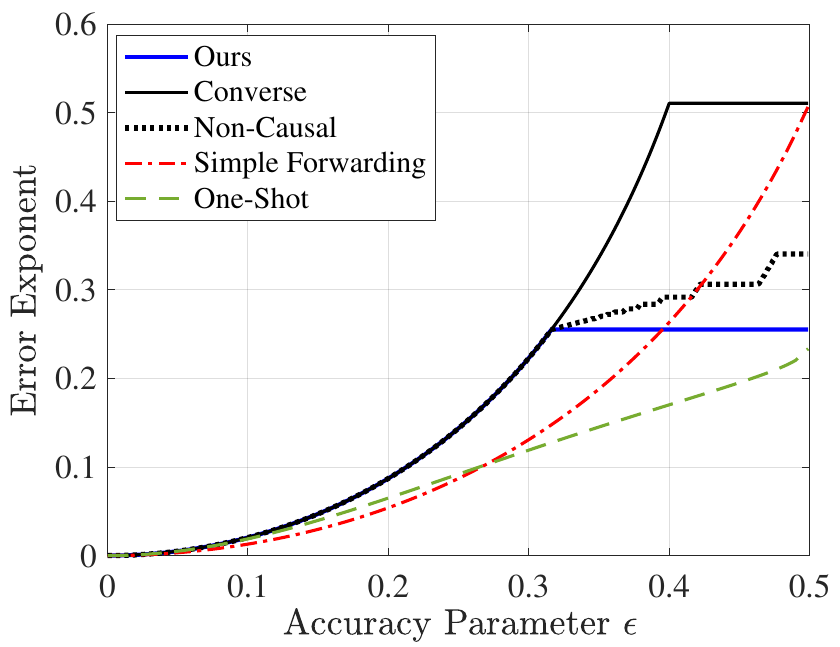}
            \par
        \end{centering}
        
        \caption{Comparison of error exponents in the Bernoulli+BSC setting with crossover probability $p = 0.1$. 
 (The jagged behavior of the curve `Non-Causal' for higher $\epsilon$ comes from requiring an integer number of codewords and that integer becoming small.) \label{fig:numerical1}}
    \end{figure}

    \begin{figure}
        \begin{centering}
            \includegraphics[width=0.4\columnwidth]{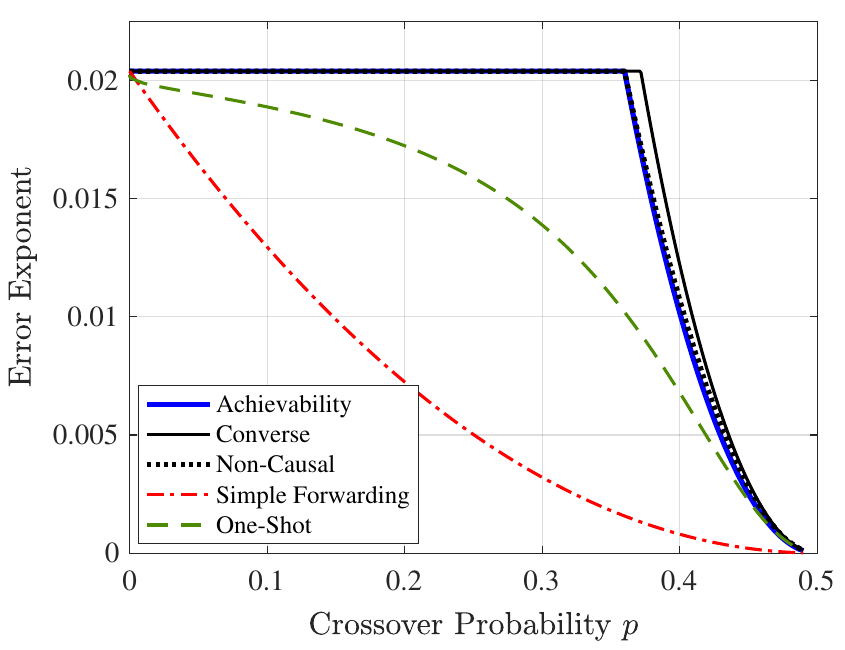}
            \par
        \end{centering}
        
        \caption{Comparison of error exponents in the Bernoulli+BSC setting with accuracy parameter $\epsilon = 0.1$.\label{fig:numerical2}}
    \end{figure}

\subsection{Description of the protocol} \label{sec:bern_protocol}

\subsubsection{Block structure}
    As noted in the introduction, we will use a block-structured design in the same way as the 1-bit learning setup of \cite{teachlearn}, but with differing details of what is done within each block.
    
    In more detail, the transmission protocol of length $n$ is broken down into $n/k$ blocks, each of size $k$.\footnote{To simplify notation we assume that $n$ is a multiple of $k$; if not, we can round down to the nearest multiple of $k$ by ignoring at most $k$ symbols, and this only does not impact the final result since we will later set $k = o(n)$.}  For each $j=1,2,\ldots, \frac{n}{k}-1$, we let $( W_{ik+1}, W_{ik+2}, \ldots, W_{(i+1)k})$ be a (deterministic) function of $(X_{(i-1)k+1}, X_{(i-1)k+2}, \dotsc, X_{ik})$. The first $k$ symbols $Z_1, Z_2, \ldots, Z_k$ are ignored by the student, and the last $k$ samples $X_{n-k+1},X_{n-k+2},\dotsc,X_n$ are ignored by the teacher.
	In other words, the $j$-th block received by the teacher only depends on the $(j-1)$-th block of samples received by the student; see Figure \ref{block-protocol} for an illustration. To simplify notation, we define
	\begin{equation}
	\vec{Z}^j = (Z_{jk+1}, Z_{jk+2}, \ldots, Z_{(j+1)k}) 
	\end{equation}
	to be the $j$-th block received by student (excluding the ignored one), and let $\vec{Z}$ refer to a generic block.

	\begin{figure}
    \centering
	\begin{tikzpicture}
		[line cap=round,line join=round,>=triangle 45,x=1.0cm,y=1.0cm, scale=0.6]
		\clip(-1,-1) rectangle (16, 7);
		\draw [line width=1.pt] (0.,0.)-- (12.,0.);
		\draw [line width=1.pt] (0.,1.)-- (12.,1.);
		\draw [line width=1.pt] (0.,3.)-- (12.,3.);
		\draw [line width=1.pt] (12.,1.)-- (12.,0.);
		\draw [line width=1.pt] (0.,1.)-- (0.,0.);
		\draw [line width=1.pt] (4.,1.)-- (4.,0.);
		\draw [line width=1.pt] (8.,1.)-- (8.,0.);
		\draw [line width=1.pt] (0.,4.)-- (12.,4.);
		\draw [line width=1.pt] (12.,3.)-- (12.,4.);
		\draw [line width=1.pt] (8.,3.)-- (8.,4.);
		\draw [line width=1.pt] (4.,3.)-- (4.,4.);
		\draw [line width=1.pt] (0.,3.)-- (0.,4.);
		\draw [->,line width=1.pt] (0,5) -- (13,5);
		\draw (14,5) node{time};
		\draw (0, 4.5) node{1};
		\draw (3.8, 4.5) node{$k$};
		\draw (7.7, 4.5) node{$2k$};
		\draw (11.7, 4.5) node{$3k$};
		\draw (0, -0.5) node{1};
		\draw (3.8, -0.5) node{$k$};
		\draw (7.7, -0.5) node{$2k$};
		\draw (11.7, -0.5) node{$3k$};
		\draw [->,line width=1.pt] (2,3) -- (6,1);
		\draw [->,line width=1.pt] (6,3) -- (10,1);
		\draw [fill=red] (0,3) rectangle (4,4);
		\draw [fill=red] (4,0) rectangle (8,1);
		\draw [fill=blue] (4,3) rectangle (8,4);
		\draw [fill=blue] (8,0) rectangle (12,1);
		\draw (14,3.5) node{teacher};
		\draw (14,0.5) node{student};
	\end{tikzpicture}
	
	\caption{A diagrammatic representation of the block protocol}
	\label{block-protocol}
	\end{figure}
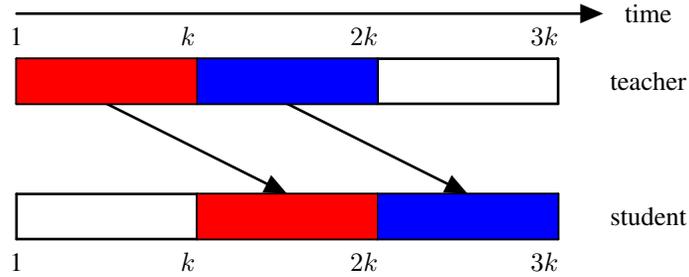

\subsubsection{Teacher strategy}
The teacher and student first agree on a length-$k$ codebook over BSC($p$) with $k+1$ codewords such that any two codewords differ in at least $\frac{k}{2}-o(k)$ bits (see Lemma \ref{lem:bsc_hamming_distance}). Let the codewords be $\vec{W}_0, \vec{W}_1, \ldots, \vec{W}_k$. 

The teacher reads in blocks of length $k$ as described above. For each block, the teacher counts the number of 1s and sends the corresponding message of the codebook.  That is, when the number of 1s received in a block is $\alpha \in \{0,1,\dotsc,k\}$, the corresponding codeword is $\vec{W}_\alpha$.

\subsubsection{Student strategy}
Let $\vec{Z}$ represent a generic block of length $k$ received by the student, and define
\begin{equation}
    g(\theta, \alpha, \vec{Z}) = \sqrt{\p(\alpha|\theta)\p(\vec{Z}| \vec{W}_\alpha)}
    \label{eq:g_def}
\end{equation}
and
\begin{equation}
    f(\theta, \vec{Z}) = \sum_{\alpha} g(\theta, \alpha, \vec{Z}).
\end{equation}
Intuitively, $f$ behaves like a likelihood function up to polynomial pre-factors; this is because if the summation over $\alpha$ were \emph{inside} the square root in \eqref{eq:g_def}, we would precisely have the square root of the likelihood.  Defining the summation to be outside the square root turns out to be more convenient for the analysis.

Recall that $\vec{Z}^j$ is the $j$-th block received by the student. Let $T$ be the set of integer multiples of $\frac1{2n}$ in $[0,1]$, and define the set 
\begin{equation}
    S = \bigg\{ \theta \in T \,\Big|\, \prod_{j=1}^{n/k-1} f(\theta, \vec{Z}^j) > \prod_{j=1}^{n/k-1} f(\theta', \vec{Z}^j), ~ \forall \theta' \in T\text{ s.t. } |\theta-\theta'|\geq 2\epsilon - \frac1{n}\bigg\}. \label{eq:setS}
\end{equation}
The student then outputs the final estimate $\hat{\theta}_n = \frac12 (\min S + \max S)$.  Intuitively, the $\theta$ values in $S$ are those that ``beat'' all $2\epsilon$-far $\theta'$ values in a binary hypothesis test when $f$ is used for the decision rule.

\subsection{Proof of Theorem \ref{thm:bernoulli} (Bernoulli+BSC achievability)} \label{sec:pf_bern}

By definition, $S$ is contained within an interval of $2 \epsilon - \frac1{n}$; this is because if some $\theta \in S$ ``beats'' some $\theta'$ then the opposite is automatically false. Let $\theta^{**}$ be $\theta^*$ rounded to an integer multiple of $\frac{1}{2n}$, where we round down if $\theta^* > \frac12$ and round up if $\theta^* < \frac12$ (i.e. we always round towards $\frac12$). Observe that if $\theta^{**} \in S$, then $\theta^{**} \in [\min S, \max S]$ and therefore the choice $\hat{\theta}_n = \frac12 (\min S + \max S)$ ensures that $|\theta^{**} - \hat{\theta}_n | \leq \frac{1}{2}\big(2\epsilon - \frac{1}{n} \big) = \epsilon - \frac{1}{2n}$, which implies $|\theta^{*} - \hat{\theta}_n | \leq \epsilon$ by the triangle inequality.  Therefore, the error probability (for $\epsilon$-deviations) is upper bounded as follows:
\begin{equation}
    \p(|\theta^{*} - \hat{\theta}_n | > \epsilon) \le \p(\theta^{**} \notin S). \label{eq:err_prob}
\end{equation}
 Towards bounding $\p(\theta^{**} \notin S)$, we will first show that for any $\theta'$ satisfying $|\theta^* - \theta'| \geq 2\epsilon$, it holds with suitably high probability that $\prod_j f(\theta^*, \vec{Z}^j) > 2 \prod_j f(\theta', \vec{Z}^j)$.  In Lemma \ref{lem:round_n} below, we will consider the effect of discretization and show that $f(\theta^{**},\vec{Z})$ does not differ too much from $f(\theta^*,\vec{Z})$.
 
\begin{lemma} \label{lem:Ef_Bern}
    Let $\theta^*$ be the true parameter and $\theta'$ be some other value. Then
    \begin{equation}
        \mathbb{E}\bigg( \frac{f(\theta', \vec{Z})}{f(\theta^*, \vec{Z})}\bigg) \leq \exp\bigg(-(k-o(k)) \cdot \min\Big(\db(\theta^*, \theta'), \frac12 D\big(1/2 \big\| p\big) \Big)\bigg).
    \end{equation}
\end{lemma}
\begin{proof}
    By the law of total probability and the definition of $f$, we have
    \begin{align}
        \mathbb{E}\bigg( \frac{f(\theta', \vec{Z})}{f(\theta^*, \vec{Z})}\bigg) 
            &= \sum_{\alpha} \p(\alpha|\theta^*) \mathbb{E} \bigg( \frac{f(\theta', \vec{Z})}{f(\theta^*, \vec{Z})}\Big| \alpha\bigg) \\
            & \leq  \sum_{\alpha} \p(\alpha|\theta^*) \mathbb{E} \bigg( \frac{f(\theta', \vec{Z})}{g(\theta^*,\alpha, \vec{Z})}\bigg| \alpha\bigg) \\
            & =\sum_{\alpha} \sum_{\alpha'}\p(\alpha|\theta^*) \mathbb{E} \bigg( \frac{g(\theta',\alpha', \vec{Z})}{g(\theta^*,\alpha, \vec{Z})}\bigg| \alpha\bigg).
        \label{eq:split_alpha_bsc}
    \end{align}
    We split the above sum into $\alpha=\alpha'$ and $\alpha \neq \alpha'$:
    \begin{itemize}
        \item When $\alpha = \alpha'$, the cancellation of $\p(\vec{Z}|\vec{W}_\alpha)$ terms in \eqref{eq:g_def} leads to the following:
    \begin{equation}
        \frac{g(\theta',\alpha, \vec{Z})}{g(\theta^*,\alpha, \vec{Z})} = \sqrt{\frac{\p(\alpha|\theta')}{\p(\alpha|\theta^*)}},
    \end{equation}
    which implies
    \begin{equation}
        \sum_\alpha \p(\alpha|\theta^*)\frac{g(\theta',\alpha, \vec{Z})}{g(\theta^*,\alpha, \vec{Z})} = 
        \sum_\alpha \p(\alpha|\theta^*)  \sqrt{\frac{\p(\alpha|\theta')}{\p(\alpha|\theta^*)}} = \exp(-k\cdot \db(\theta, \theta')).
        \label{eq:alpha_equal_bsc}
    \end{equation}
    \item When $\alpha \neq \alpha'$, substituting the definition of $g$ and applying $\p(\alpha|\theta^*) \le 1$ gives
    \begin{equation}
         \p(\alpha|\theta^*) \frac{g(\theta', \alpha', \vec{Z})}{g(\theta^*, \alpha, \vec{Z})} \leq \p(\alpha|\theta^*) \sqrt\frac{\p(\vec{Z}|\vec{W}_{\alpha'})} {\p(\alpha|\theta^*)\p(\vec{Z}|\vec{W}_{\alpha})} \leq \sqrt\frac{\p(\vec{Z}|\vec{W}_{\alpha'})} {\p(\vec{Z}|\vec{W}_{\alpha})},
    \end{equation}
    and taking the expectation given $\alpha$, we obtain
    \begin{equation}
        \mathbb{E}\left(\sqrt\frac{\p(\vec{Z}|\vec{W}_{\alpha'})}{\p(\vec{Z}|\vec{W}_{\alpha})}\Big| \alpha\right) = \sum_{\vec{Z}} \p(\vec{Z}|\vec{W}_\alpha) \sqrt\frac{\p(\vec{Z}|\vec{W}_{\alpha'})}{\p(\vec{Z}|\vec{W}_{\alpha})} = \sum_{\vec{Z}} \sqrt{\p(\vec{Z}|W_{\alpha}) \cdot \p(\vec{Z}|W_{\alpha})} = \rho(\vec{W}_{\alpha}, \vec{W}_{\alpha'},P_{Z|W}).
    \end{equation}
    By Lemma \ref{lem:bsc_hamming_distance}, $W_{\alpha}$ and $W_{\alpha'}$ have Hamming distance $\frac k2 - o(k)$, which gives % $\rho(\vec{W}_\alpha, \vec{W}_{\alpha'},P_{Z|W}) \leq \exp(-(k+o(k))\frac12 D(1/2 \| p))$ and hence
    \begin{equation}
        \sum_{(\alpha,\alpha') \,:\, \alpha \neq\alpha'}\p(\alpha|\theta^*) \mathbb{E} \bigg( \frac{g(\theta',\alpha', \vec{Z})}{g(\theta^*,\alpha, \vec{Z})}\big| \alpha\bigg) \leq \sum_{(\alpha,\alpha') \,:\, \alpha \neq \alpha'} \rho(\vec{W}_\alpha, \vec{W}_{\alpha'},P_{Z|W})  \leq \exp\Big(-(k-o(k))\frac12 D(1/2 \| p)\Big),
        \label{eq:alpha_not_equal_bsc}
    \end{equation}
    where the last step follows from tensorization (Lemma \ref{lem:tens_db}) and the identity between $\db$ and KL divergence (Lemma \ref{lem:eps_equiv} with $\epsilon = \frac{1}{2}-p$).
    \end{itemize}
    Substituting \eqref{eq:alpha_equal_bsc} and \eqref{eq:alpha_not_equal_bsc} into \eqref{eq:split_alpha_bsc} gives the required result.
\end{proof}

Since the $n/k-1$ blocks are independent, Lemma \ref{lem:Ef_Bern} implies the following when $k$ is chosen such that $k \to \infty$ and $n/k \to \infty$ (e.g., $k = \sqrt{n}$):
\begin{equation}
    \mathbb{E} \bigg(\prod_j \frac{f(\theta',\vec{Z}^j)}{f(\theta^*, \vec{Z}^j)} \bigg)=\prod_j \mathbb{E} \bigg(\frac{f(\theta',\vec{Z}^j)}{f(\theta^*, \vec{Z}^j)} \bigg) \leq \exp(-(n-o(n))\cdot \min\bigg(\db(\theta^*, \theta'), \frac12 D\big(1/2 \big\| p\big)\bigg).
\end{equation}

% If $\theta'$ came from a finite set, we could immediately proceed to apply Markov's inequality and a union bound.  However, $\theta'$ can be any real value in $[0,1]$ satisfying $|\theta - \theta'| > 2\epsilon$. To address this, we discretize $\theta'$ and apply a union bound only to rounded values via the following lemma.

Next, we provide the following lemma relating the true value $\theta^* \in [0,1]$ to its rounded version $\theta^{**} \in T$.

\begin{lemma} For all possible $(\vec{Z}^1, \vec{Z}^2, \ldots, \vec{Z}^{n/k-1})$, it holds that $\prod_j \frac{f(\theta^*, \vec{Z}^j)}{f(\theta^{**}, \vec{Z}^j)} \leq 2$.
\label{lem:round_n}
\end{lemma}
\begin{proof} 
Recall that $\theta^{**}$ is defined by rounding $\theta^*$ to the nearest multiple of $\frac{1}{2n}$ towards $\frac{1}{2}$.
We focus on the case that $\theta^{*}\geq\theta^{**}\geq 1/2$; the case $\theta^* \le \theta^{**} \le 1/2$ is analogous and is omitted.  We have
\begin{equation}
    \frac{\p(\alpha|\theta^{*})}{\p(\alpha|\theta^{**})} \leq \left(\frac{\theta^{*}}{\theta^{**}}\right)^{k} \leq \left(\frac{\theta^{**} + \frac{1}{2n}}{\theta^{**}}\right)^{k} \leq \Big(1+\frac1n\Big)^k \leq \exp(k/n),
\end{equation}
where we respectively used that $\alpha$ is a sum of $k$ i.i.d.~Bernoulli random variables, the assumption $|\theta^* - \theta^{**}| < \frac1{2n}$, the assumption $\theta^* \ge \frac{1}{2}$, and the identity $1+a \le e^a$.  It follows that
\begin{equation}
    f(\theta^{*}, \vec{Z}) = \sum_{\alpha} \sqrt{\p(\alpha|\theta^{*}) \p(\vec{Z}|\vec{W}_{\alpha})} \leq \exp(k/(2n)) \sum_{\alpha} \sqrt{\p(\alpha|\theta^{**}) \p(\vec{Z}|\vec{W}_{\alpha})} = \exp(k/(2n)) f(\theta^{**},\vec{Z}),
\end{equation}
and hence
\begin{equation}
    \prod_j \frac{f(\theta^{*}, \vec{Z}^j)}{f(\theta^{**}, \vec{Z}^j)} \leq \exp(k/(2n))^{n/k-1} \le \exp(1/2) \leq 2.
\end{equation}
\end{proof}

Recall that $T$ is the set of all integer multiples of $\frac1{2n}$ in $[0,1]$, and that $S$ is defined in \eqref{eq:setS}.  It follows that if we define
\begin{equation}
    T_{\epsilon} = T \setminus \Big[\theta^{**}-2\epsilon+\frac1n, \theta^{**}+2\epsilon-\frac1n\Big], 
\end{equation}
then the condition $\theta^** \notin S$ implies the existence of some $\theta' \in T_{\epsilon}$ such that $\prod_j f(\theta', \vec{Z}^j) \geq \prod_j f(\theta^{**},\vec{Z}^j)$.  This further implies via Lemma \ref{lem:round_n} that $2 \prod_j f(\theta',\vec{Z}^j) \geq 2\prod_j f(\theta^{**}, \vec{Z}^j)\geq \prod_j f(\theta^*,\vec{Z}^j)$.  In addition, since $|\theta^{**}-\theta^*|\leq \frac1{2n}$ and $|\theta'-\theta^{**}|\geq 2\epsilon-\frac1n$, we have $|\theta^{'}-\theta^*| \geq 2\epsilon-\frac3{2n}$ due to the triangle inequality.

Combining the above observations, we have
\begin{eqnarray}
    \p(\theta^{**} \notin S) & = &\p \Bigg(\bigcup_{\theta' \in T_{\epsilon}} \bigg\{ \prod_j f(\theta',\vec{Z}^j) \ge \prod_j f(\theta^{**},\vec{Z}^j) \bigg\}\Bigg) \nonumber \\
    &\leq& \p \Bigg(\bigcup_{\theta' \in T_{\epsilon}} \bigg\{2\prod_j f(\theta'. \vec{Z}^j) \ge \prod_j f(\theta^{*},\vec{Z}^j)\bigg\}\Bigg) \label{eq:setT_step1}\\
    & \leq& \sum_{\theta' \in T_{\epsilon}} \p \Bigg( 2\prod_j f(\theta'. \vec{Z}^j) \ge \prod_j f(\theta^{*},\vec{Z}^j)\Bigg) \label{eq:setT_step2}\\
    & \leq & 2\sum_{\theta' \in T_{\epsilon}} \mathbb{E} \Bigg( \prod_j \frac{f(\theta',\vec{Z}^j)}{f(\theta^*, \vec{Z}^j)}\Bigg) \label{eq:setT_step3}\\
    & \leq &\exp(-(n+o(n))\cdot \min\bigg(\min_{|\theta'-\theta^*| \ge 2\epsilon - \frac{3}{2n}}\db(\theta^*, \theta'), \frac12 D\big(1/2 \big\| p\big)\bigg), \label{eq:setT_step4}
\end{eqnarray}
where \eqref{eq:setT_step1} holds as discussed above, \eqref{eq:setT_step2} uses a union bound, \eqref{eq:setT_step3} applies Markov's inequality, and \eqref{eq:setT_step4} applies Lemma \ref{lem:Ef_Bern} and the above-established inequality $|\theta^{'}-\theta^*| \geq 2\epsilon-\frac3{2n}$.

By \eqref{eq:err_prob}, \eqref{eq:setT_step4}, and the continuity of $\db(\theta_1,\theta_2)$, we deduce that it is possible to achieve an error exponent of 
\begin{equation}
    E \geq \min\bigg(\min_{|\theta'-\theta^*|\geq 2\epsilon}\db(\theta^*, \theta'), \frac12 D\big(1/2 \big\| p\big) \bigg),
    \label{eq:e_theta}
\end{equation}
and it remains to simplify the first term.  This is done in the following lemma, which we prove in Appendix \ref{app:pf_simplification_Bern} using some elementary calculus.

\begin{lemma} \label{lem:simplification_Bern}
    For all $\theta'$ and $\theta^*$ satisfying $|\theta'-\theta^*|\geq 2\epsilon$, it holds that $\db(\theta^*, \theta') \ge D\big(\frac{1}{2} \big\| \frac{1}{2} + \epsilon\big)$.
\end{lemma}

Combining this lemma with \eqref{eq:e_theta} yields Theorem \ref{thm:bernoulli}.

\subsection{Proof of Corollary \ref{cor:bern_conv} (Converse for non-causal teacher)} \label{sec:pf_bern_conv}

We apply the general converse result for non-causal protocols from Lemma \ref{lem:noncausal}, and then rewrite the channel term using $c_M$ from Definition \ref{def:cM} and its characterization for the BSC from Lemma \ref{lem:cM}.  It then only remains to show that the source term simplifies as $E^{\rm src}_{\epsilon} = D\big(\frac{1}{2} \| \frac{1}{2} + \epsilon\big)$. 
It follows immediately from Theorem \ref{thm:bernoulli} (e.g., paired with a noiseless binary channel) that $E^{\rm src}_{\epsilon} \ge D\big(\frac{1}{2} \| \frac{1}{2} + \epsilon\big)$, so it suffices to show a matching converse, i.e., $E^{\rm src}_{\epsilon} \le D\big(\frac{1}{2} \| \frac{1}{2} + \epsilon\big)$.

To see this, we simply observe that attaining accuracy $\epsilon$ implies being able to successfully perform a hypothesis test between two Bernoulli distributions with means $\frac12+\epsilon'$ and $\frac12-\epsilon'$ for any $\epsilon' > \epsilon$.   It is well-known (e.g., see \cite{berlekampI}) that the optimal error exponent for distinguishing these is $\db(\frac12-\epsilon', \frac12 + \epsilon')$, which equals $D(\frac12 \| \frac12 + \epsilon')$ by Lemma \ref{lem:eps_equiv}.  Taking $\epsilon' \to \epsilon$ from above gives the desired result.

\section{Sub-Gaussian Source and Arbitrary Discrete Memoryless Channel} \label{sec:subgaussian}

\subsection{Problem statement}

We now turn to a more general setting where $X$ is not necessarily known to lie in any parametric family, but is instead only known to be \emph{sub-Gaussian} with parameter $\sigma^2 > 0$ (after centering), i.e.,
\begin{equation}
    \mathbb{E}[\exp(s(X-\theta^*))] \leq \exp\Big(\frac12 \sigma^2 s^2\Big) ~~\forall s \in \mathbb{R},
\end{equation}
where $\theta^* = \EE[X]$.  
Note that $\sigma^2$ may differ from ${\rm Var}[X]$. 

We also generalize the communication channel to be an arbitrary discrete memoryless channel $P_{Z|W}$.\footnote{This generalization can also be incorporated into Section \ref{sec:bernoulli} in the same way.}  Our goal remains to characterize the error exponent associated with the event $|\hat{\theta}_n - \theta^*| > \epsilon$ for some $\epsilon > 0$, i.e., bound the optimal exponent $E^*$ as defined in \eqref{eq:E_def}, with $\calP_X$ being the sub-Gaussian class.

In this setting, allowing arbitrary $\theta^* \in \RR$ is not possible, because the decoder can only receive at most $O(n)$ bits of information, which cannot represent an unbounded real number within any given accuracy.
To simplify the exposition, we assume that $\theta^* \in [0,1]$, but with only minor modifications we can handle an arbitrary fixed interval (known to the teacher and student) or even an interval whose width grows as $O(n^C)$ for any fixed constant $C > 0$; see Section \ref{sec:discussion} for further discussion.

Under the sub-Gaussianity assumption, if we take $n$ samples $X_1, X_2, \ldots, X_n$ and let $\bar{X} = \frac{1}{n}\sum_{i=1}^n X_i$ be the sample mean, then we have (e.g., see \cite[Sec.~2.5]{vershynin}):
\begin{equation}
    \mathbb{P}(\bar{X} \geq \theta^* + \epsilon) \leq \exp\Big(-\frac{n\epsilon^2 }{2\sigma^2}\Big), \quad \mathbb{P}(\bar{X} \leq \theta^* - \epsilon) \leq \exp\Big(-\frac{n\epsilon^2 }{2\sigma^2}\Big). \label{eq:subg_bound}
\end{equation}
Hence, we have an error exponent of $\frac{\epsilon^2}{2\sigma^2}$ when the samples are observed directly.  We will shortly argue that this cannot be improved in general, by specializing to the Gaussian case.

Our main result in this section is the following.
 
\begin{theorem} \label{thm:subg_ach}
    {\em (Sub-Gaussian Achievability)}
    For sub-Gaussian sources $P_X$ with mean in $[0,1]$ and sub-Gaussianity parameter $\sigma^2 > 0$, any discrete memoryless channel $P_{Z|W}$, and any accuracy parameter $\epsilon \in \big(0,\frac{1}{2}\big)$, the following error exponent is achievable:
    \begin{equation}
        E^* \ge \min\Big(\frac{\epsilon^2}{2\sigma^2}, E^{\rm chan}(0)\Big). \label{eq:subg_ach}
    \end{equation}
\end{theorem}

To understand the tightness of this result, we state the following corollary of Lemma \ref{lem:noncausal} for Gaussian sources (which is a special case of sub-Gaussian).

\begin{cor} \label{cor:subg_conv}
    {\em (Gaussian Converse)}
    For Gaussian sources $P_X$ with mean in $[0,1]$ and variance $\sigma^2 > 0$, under any discrete memoryless channel $P_{Z|W}$, and any accuracy parameter $\epsilon \in \big(0,\frac{1}{2}\big)$, it holds (even in the non-causal setting) that
    \begin{equation}
        E^* \le \min\Big( \frac{\epsilon^2}{2\sigma^2}, c_{\lceil 1/(2\epsilon) \rceil} E^{\rm chan}(0)\Big), \label{eq:subg_conv}
    \end{equation}
    where $c_{\lceil 1/(2\epsilon) \rceil}$ is defined in Definition \ref{def:cM}, and satisfies $c_{\lceil 1/(2\epsilon) \rceil} \to 1$ as $\epsilon \to 0$ (by Lemma \ref{lem:cM}).
\end{cor}

Thus, the upper and lower bounds match whenever the $\min\{\cdot,\cdot\}$ in \eqref{eq:subg_ach} is achieved by the first term.   Moreover, even the second terms match to within a factor that approaches 1 as $\epsilon \to 0$.  (Unlike the BSC, we do not have any simple exact expression or precise convergence rate.)
% In addition, we have $E^{\rm chan}_{\lfloor 1/(2\epsilon)\rfloor} \to E^{\rm chan}(0)$ as $\epsilon \to 0$, so the second terms also become increasingly close as $\epsilon \to 0$.  However, unlike the BSC case, we do not make any claim about the rate at which their ratio approaches one. 

We will prove Theorem \ref{thm:subg_ach} in Sections \ref{sec:bern_protocol} and \ref{sec:pf_bern}, and we will prove Corollary \ref{cor:subg_conv} in Section \ref{sec:pf_subg_conv}.

\subsection{Description of the protocol} \label{sec:subg_protocol}

We use the same block-structured approach as the one described in Section \ref{sec:bernoulli} for some block length $k$ satisfying $k \to \infty$ and $n/k \to \infty$, but with differing details described below. 
Before proceeding, we claim that it suffices to handle the case that $X$ satisfies the following assumption.

\begin{assump} \label{as:rounding}
    The distribution $P_X$ is such that $X$ is a multiple of $1/k$ with probability one.
\end{assump}

We proceed to argue that proving Theorem \ref{thm:subg_ach} under this assumption readily implies the general case.  To see this, suppose that in the general case, the teacher rounds $X$ as follows:
\begin{itemize}
    \item Let $X^{\uparrow}$ (respectively, $X^{\downarrow}$) equal $X$ rounded up (respectively, down) to the nearest multiple of $1/k$;
    \item Round to $X^{\uparrow}$ or $X^{\downarrow}$ with probability proportional to the relative position of $X$ in the interval $[X^{\downarrow},X^{\uparrow}]$, i.e., the associated probabilities are $\frac{X-X^{\downarrow}}{X^{\uparrow}-X^{\downarrow}}$ and $\frac{X^{\uparrow}-X}{X^{\uparrow}-X^{\downarrow}}$  respectively.
\end{itemize}
A direct calculation of the conditional mean (given $X$) shows the ``rounding noise'' has mean zero, and thus the rounded random variable still has mean $\theta^*$. (Note that this is often called \emph{stochastic quantization} and is not a new idea.)  However, the impact of the rounding on the sub-Gaussianity parameter $\sigma^2$ needs to be checked carefully.

Let $\Delta$ be the difference between the rounded value and the original value of $X$.  We characterize the sub-Gaussianity of the rounded random variable $X+\Delta$ as follows:
\begin{itemize}
    \item By assumption, $X$ is sub-Gaussian with parameter $\sigma^2$.
    \item We can use the boundedness of $\Delta$ to deduce that it is sub-Gaussian; specifically, since $|\Delta| \leq \frac1k$ and $\mathbb{E}[\Delta] = 0$, we have that $\Delta$ is sub-Gaussian with parameter $\frac 1{k^2}$ \cite[Corollary 2.6]{subgaussian}.
    \item The sub-Gaussianity of both $X$ and $X'$ implies the same for $X+X'$, but since the two are not independent we cannot simply sum their sub-Gaussianity parameters.  Rather, a more general property (based on H\"older's inequality) not requiring independence gives an overall sub-Gaussian parameter of $(\sigma + \frac1k)^2$ \cite[Theorem 2.7]{subgaussian}.\footnote{More generally, the property in \cite[Theorem 2.7]{subgaussian} is that if $X_1$ and $X_2$ have sub-Gaussianity parameters $\sigma_1^2$ and $\sigma_2^2$, then $X_1+X_2$ is sub-Gaussian with parameter $(\sigma_1 + \sigma_2)^2$.  
    In \cite{subgaussian} this result is attributed to \cite{buldygin1980sub}, but \cite{subgaussian} also has a self-contained proof.  If $X_1$ and $X_2$ are independent, the parameter further improves to $\sigma_1^2 + \sigma_2^2$.}
\end{itemize}
Our final error exponent in Theorem \ref{thm:subg_ach} is continuous in $\sigma$, so the change from $\sigma^2$ to $(\sigma + \frac1k)^2$ (with $k \to \infty$) does not impact the final result. 
In the remainder of the section, we adopt Assumption \ref{as:rounding}.

\subsubsection{Teacher strategy}

Let the sample mean be $\bar{X} = \frac{1}{n}\sum_{i=1}^n X_i$, and define $\alpha$ as follows:
\begin{equation}
    \alpha = \begin{cases}
        k & \bar{X} > k\\
        -k & \bar{X} < k\\
        \bar{X} & k \leq \bar{X} \leq k.
    \end{cases} \label{eq:alpha_def}
\end{equation}
Under Assumption \ref{as:rounding}, $\bar{X}$ is a multiple of $\frac1{k^2}$. Since $\alpha$ is simply given by $\bar{X}$ clipped to $[-k,k]$, it follows that $\alpha$ can take on one of at most $2k^3 + 1 \le 4k^3$ values.  Accordingly, we consider a set of at most $4k^3$ codewords of length $k$ for the channel $P_{Z|W}$ satisfying  Corollary \ref{cor:db_dmc}. The teacher reads in blocks of $k$, and transmits the codeword corresponding to the sample mean. If the sample mean is $\alpha$, we let $\vec{W}_{\alpha}$ be the corresponding codeword sent by the teacher.

We claim that $\alpha$ satisfies similar sub-Gaussian style bounds as the sample mean $\bar{X}$ (see \eqref{eq:subg_bound}); specifically, the bound that will be useful for our purposes is
\begin{equation}
    \p(\alpha|\theta^*) \leq \exp\left( -\frac{k(\theta^*-\alpha)^2}{2\sigma^2}\right),
    \label{eq:alpha_tail}
\end{equation}
where here and subsequently, $\p(\alpha|\theta^*)$ represents the PMF of $\alpha$ under an arbitrary distribution $P_X$ that (i) lies in the sub-Gaussian class $\mathcal{P}_X$, (ii) satisfies Assumption \ref{as:rounding}, and (iii) has mean $\theta^*$.  
To see that \eqref{eq:alpha_tail} holds, we consider the three cases in \eqref{eq:alpha_def}.  Firstly, if $\alpha \notin \{k, -k\}$, then $\bar{X} = \alpha$ and the tail bound for $\bar{X}$ trivially implies the same for $\alpha$.  On the other hand, if $\alpha = k$, then using $\theta^* < \alpha$ (recall that $\theta^* \in [0,1]$) and $\bar{X} > \alpha = k$ gives
\begin{equation}
    \mathbb{P}(\alpha|\theta^*) = \mathbb{P}(\bar{X}\geq k|\theta^*)  \leq \exp\left( -\frac{k(\theta^*-k)^2}{2\sigma^2}\right) = \exp\left( -\frac{k(\theta^*-\alpha)^2}{2\sigma^2}\right).
\end{equation}
The case $\alpha = -k$ follows from an analogous argument. 

\subsubsection{Student strategy}  Recall that $\vec{Z}$ denotes a generic length-$k$ block received by the student.  
We wish to take a similar approach as the Bernoulli + BSC case, for some function $g(\theta, \alpha, \vec{Z})$ defined similarly to \eqref{eq:g_def}. Since we are not given the probability distribution $\p(\alpha | \theta)$, we instead use $\exp(-\frac{k(\theta-\alpha)^2}{2\sigma^2})$ to align with the sub-Gaussianity of $X$ (see \eqref{eq:alpha_tail}).  Specifically, we define
\begin{equation}
    g(\theta, \alpha, \vec{Z}) = \exp\Big( -\frac{k(\theta-\alpha)^2}{4\sigma^2}\Big) \cdot p(\vec{Z}|\vec{W}_\alpha), \label{eq:g_subg}
\end{equation}
and
\begin{equation}
    f(\theta, \vec{Z}) = \sum_\alpha g(\theta, \alpha, \vec{Z}), \label{eq:f_subg}
\end{equation}
with $f$ serving as a ``proxy for the likelihood function'' as before.  
Furthermore, let $\vec{Z}^j$ denote the $j$-th block received by the student, and let $T$ be the set of all integer multiples of $\frac1{n^2}$ in $[0,1]$. Then, define
\begin{equation}
    S = \bigg\{ \theta \in T \,\Big|\, \prod_j f(\theta, \vec{Z}^j) > \prod_j f(\theta', \vec{Z}^j) ~ \forall \theta' \in T \text{ s.t. } |\theta-\theta'|\geq 2\epsilon - \frac2{n^2}\bigg\}.
\end{equation}
Similarly to Section \ref{sec:bernoulli}, $S$ is contained within an interval of length $2 \epsilon-\frac{2}{n^2}$, and we let the final estimate be $\hat{\theta}_n = \frac12(\min S + \max S)$.

\subsection{Proof of Theorem \ref{thm:subg_ach} (Achievability for sub-Gaussian sources)} \label{sec:pf_subg_ach}

% By definition, $S$ is contained within an interval of $2 \epsilon - \frac2{n^2}$. 
Let $\theta^{**}$ be $\theta^*$ rounded to the nearest integer multiple of $\frac{1}{n^2}$. Observe that if $\theta^{**} \in S$, then $\theta^{**} \in [\min S, \max S]$.  Hence, and using $\max S - \min S \le 2 \epsilon-\frac{2}{n^2}$ as established above, the choice $\hat{\theta}_n = \frac12 (\min S + \max S)$ ensures that $|\theta^{**} - \hat{\theta}_n | \leq \epsilon - \frac{1}{n^2}$, which implies $|\theta^{*} - \hat{\theta}_n | \leq \epsilon$ by the triangle inequality.  Therefore, the error probability is upper bounded by $\p(\theta^{**} \notin S)$. 

Towards characterizing $\p(\theta^{**} \notin S)$, we start with the following analog of Lemma \ref{lem:Ef_Bern}.

\begin{lemma} \label{lem:Ef_subg}
    Let $\theta^*$ be the true parameter, and let $\theta'$ be some other value. Then
    \begin{equation}
        \mathbb{E}\bigg( \frac{f(\theta', \vec{Z})}{f(\theta^*, \vec{Z})}\bigg) \leq \exp\bigg(-(k-o(k))\min\Big(\frac{(\theta^*-\theta')^2}{8\sigma^2}, E^{\rm chan}(0)\Big)\bigg)
    \end{equation}
\end{lemma}
\begin{proof}
    % Let $\bar{X}$ be the empirical mean of the student block $\vec{X}$ corresponding to $\vec{Z}$.  
    Recall the definition of $\alpha$ in \eqref{eq:alpha_def}, and that it takes on one of at most $4k^3$ values due to Assumption \ref{as:rounding}.  We proceed as follows:
    \begin{eqnarray}
        \mathbb{E}\bigg( \frac{f(\theta', \vec{Z})}{f(\theta^*, \vec{Z})}\bigg) &=& \sum_{\alpha} \p( \alpha | \theta^* )\mathbb{E} \bigg( \frac{f(\theta', \vec{Z})}{f(\theta^*, \vec{Z})}\Big| \alpha\bigg) \label{eq:Ef2_step1} \\
        & \leq & \sum_{\alpha} \exp\Big(-\frac{k(\theta^*-\alpha)^2}{2\sigma^2}\Big)\mathbb{E} \bigg( \frac{f(\theta', \vec{Z})}{f(\theta^*, \vec{Z})}\Big| \alpha\bigg) \label{eq:Ef2_step2}\\
        &\leq&  \sum_{\alpha} \exp\Big(-\frac{k(\theta^*-\alpha)^2}{2\sigma^2}\Big)\mathbb{E} \bigg( \frac{f(\theta', \vec{Z})}{g(\theta^*,\alpha, \vec{Z})}\bigg| \alpha\bigg) \label{eq:Ef2_step3}\\
        &=&\sum_{\alpha} \sum_{\alpha'}\exp\Big(-\frac{k(\theta^*-\alpha)^2}{2\sigma^2}\Big) \mathbb{E} \bigg( \frac{g(\theta',\alpha', \vec{Z})}{g(\theta^*,\alpha, \vec{Z})}\bigg| \alpha\bigg), 
        \label{eq:split_alpha_dmc}
    \end{eqnarray}
    where \eqref{eq:Ef2_step1} uses the law of total expectation, \eqref{eq:Ef2_step2} uses the sub-Gaussian concentration property (see \eqref{eq:alpha_tail}), and \eqref{eq:Ef2_step3}--\eqref{eq:split_alpha_dmc} use the definition of $f$.
    
    We split the above sum into the cases $\alpha=\alpha'$ and $\alpha \neq \alpha'$:
    \begin{itemize}
        \item When $\alpha = \alpha'$, the $\p(\vec{Z}|\vec{W}_\alpha)$ terms from \eqref{eq:g_subg} cancel out, and we obtain
    \begin{equation}
        \exp\Big(-\frac{k(\theta^*-\alpha)^2}{2\sigma^2}\Big)\frac{g(\theta',\alpha, \vec{Z})}{g(\theta^*,\alpha, \vec{Z})} = \exp\Big( -\frac{k(\theta'-\alpha)^2}{4\sigma^2}- \frac{k(\theta^*-\alpha)^2}{4\sigma^2}\Big)
    \end{equation}
    By a simple differentiation exercise, the maximum value of $-\frac{k(\theta'-\alpha)^2}{4\sigma^2}- \frac{k(\theta^*-\alpha)^2}{4\sigma^2}$ 
    occurs at $\alpha = \frac12 (\theta^* + \theta')$ and achieves the value $-\frac{k(\theta' - \theta^*)}{8\sigma^2}$.  Hence, and recalling that there are at most $4k^3$ values of $\alpha$, we obtain
    \begin{equation}
        \sum_\alpha  \exp\Big(-\frac{k(\theta^*-\alpha)^2}{2\sigma^2}\Big)\frac{g(\theta',\alpha, \vec{Z})}{g(\theta^*,\alpha, \vec{Z})} \leq 4k^3 \cdot \exp\bigg(-\frac{k(\theta' - \theta^*)^2}{8\sigma^2}\bigg).
        \label{eq:alpha_equal_dmc}
    \end{equation}
    \item When $\alpha \neq \alpha'$, we substitute the definition of $g$ to obtain
    \begin{equation}
          \exp\Big(-\frac{k(\theta^*-\alpha)^2}{2\sigma^2}\Big) \frac{g(\theta', \alpha', \vec{Z})}{g(\theta^*, \alpha, \vec{Z})} = \exp\Big(-\frac{k(\theta^*-\alpha)^2}{4\sigma^2}\Big) \cdot \exp\Big(-\frac{k(\theta'-\alpha)^2}{4\sigma^2}\Big)\sqrt\frac{\p(\vec{Z}|\vec{W}_{\alpha'})} {\p(\vec{Z}|\vec{W}_{\alpha})} \leq \sqrt\frac{\p(\vec{Z}|\vec{W}_{\alpha'})} {\p(\vec{Z}|\vec{W}_{\alpha})} \label{eq:neq_step1_subg}
    \end{equation}
    and taking the expectation given $\alpha$, we obtain
    \begin{equation}
        \mathbb{E}\bigg(\sqrt\frac{\p(\vec{Z}|\vec{W}_{\alpha'})}{\p(\vec{Z}|\vec{W}_{\alpha})}\Big| \alpha\bigg) = \sum_{\vec{Z}} \p(\vec{Z}|W_\alpha) \sqrt\frac{\p(\vec{Z}|\vec{W}_{\alpha'})}{\p(\vec{Z}|\vec{W}_{\alpha})} = \sum_{\vec{Z}} \sqrt{\p(\vec{Z}|W_{\alpha}) \cdot \p(\vec{Z}|W_{\alpha})} = \rho(\vec{W}_{\alpha}, \vec{W}_{\alpha'},P_{Z|W}). \label{eq:neq_step2_subg}
    \end{equation}
    We have from Corollary \ref{cor:db_dmc} that $\rho(\vec{W}_\alpha, \vec{W}_{\alpha'},P_{Z|W}) \leq \exp(-(k-o(k))E^{\rm chan}(0))$, and combining this with \eqref{eq:neq_step1_subg}--\eqref{eq:neq_step2_subg} gives
    \begin{eqnarray}
        & &\sum_{(\alpha,\alpha') \,:\, \alpha \neq\alpha'}\exp\bigg(-\frac{k(\theta^*-\alpha)^2}{2\sigma^2}\bigg) \mathbb{E} \bigg( \frac{g(\theta',\alpha', \vec{Z})}{g(\theta^*,\alpha, \vec{Z})}\bigg| \alpha\bigg) \\
        & \leq &  \sum_{(\alpha,\alpha') \,:\, \alpha \neq\alpha'}\mathbb{E}\bigg(\sqrt\frac{\p(\vec{Z}|\vec{W}_{\alpha'})}{\p(\vec{Z}|\vec{W}_{\alpha})}\Big| \alpha\bigg)\\
        &=& \sum_{(\alpha,\alpha') \,:\, \alpha \neq\alpha'}\rho(\vec{W}_\alpha, \vec{W}_{\alpha'},P_{Z|W}) \\
        &\leq& \exp(-(k-o(k))E^{\rm chan}(0)),
        \label{eq:alpha_not_equal_dmc}
    \end{eqnarray}
    \end{itemize}
    where the last step also uses that there are at most $4k^3 = {\rm poly}(k)$ values of $\alpha$.
    Substituting \eqref{eq:alpha_equal_dmc} and \eqref{eq:alpha_not_equal_dmc} into \eqref{eq:split_alpha_dmc} gives the desired result.
\end{proof}

We now move from studying a single block to the entire collection of $\frac{n}{k}-1$ blocks, and we let $k$ have an arbitrary dependence on $n$ satisfying $k \to \infty$ and $\frac{n}{k} \to \infty$ (e.g, $k = \sqrt{n}$).  We have
\begin{eqnarray}
    \mathbb{P} \Bigg(\prod_j\frac{f(\theta', \vec{Z}^j)}{f(\theta^*, \vec{Z}^j)} > \frac12 \Bigg)
    &\leq& 2 \cdot \mathbb{E} \Bigg(\prod_j\frac{f(\theta', \vec{Z}^j)}{f(\theta^*, \vec{Z}^j)} \Bigg) \label{eq:allblocks_subg_0} \\
    &=& 2 \cdot \prod_j \mathbb{E} \Bigg(\frac{f(\theta', \vec{Z}^j)}{f(\theta^*, \vec{Z}^j)} \Bigg) \label{eq:allblocks_subg_1}\\
    &\leq &\Bigg(\exp\Big( -(k-o(k))\min\Big(\frac{(\theta^*-\theta')^2}{8\sigma^2}, E^{\rm chan}(0)\Big) \Big)\Bigg) ^{n/k-1} \label{eq:allblocks_subg_2}\\
    &\leq & \exp\bigg( -(n-o(n))\min\Big(\frac{(\theta^*-\theta')^2}{8\sigma^2}, E^{\rm chan}(0)\Big)\bigg), \label{eq:allblocks_subg_3}
\end{eqnarray}
where \eqref{eq:allblocks_subg_0} applies Markov's inequality, \eqref{eq:allblocks_subg_1} uses the independence of the blocks, \eqref{eq:allblocks_subg_2} applies Lemma \ref{lem:Ef_subg}, and \eqref{eq:allblocks_subg_3} uses the above scaling assumptions on $k$.

Recalling that $\theta^{**}$ is $\theta^*$ rounded to the nearest multiple of $\frac1{n^2}$, it holds for all $\alpha$ that
\begin{equation}
    \left|\frac{k(\theta^*-\alpha)^2}{4\sigma^2} - \frac{k(\theta^{**}-\alpha)^2}{4\sigma^2} \right| \le \frac{k}{4\sigma^2}|\theta^{*}-\theta^{**}
    | \cdot |\theta^* + \theta^{**} - 2\alpha| \leq \frac{k^2}{n^2\sigma^2}, \label{eq:theta''_subg}
\end{equation}
where the first inequality uses $a^2 - b^2 = (a-b)(a+b)$, and the second inequality bounds the two absolute values by $\frac{1}{n^2}$ and $k+2 \le 4k$ respectively.  
Therefore,
\begin{eqnarray} % \Big( \frac{p}{1-p}\Big)^{d_H(\vec{Z}_\alpha, \vec{Z})/2}
    f(\theta^*, \vec{Z}) &=& \sum_\alpha\exp\Big( -\frac{k(\theta^*-\alpha)^2}{4\sigma^2}\Big) \cdot \p(\vec{Z}|\vec{W}_{\alpha}) \\
    &\leq &\exp\left(\frac{k^2}{n^2 \sigma^2}\right) \cdot \sum_\alpha\exp\Big( -\frac{k(\theta^{**}-\alpha)^2}{4\sigma^2}\Big) \cdot \p(\vec{Z}|\vec{W}_{\alpha}) \\
    &=&  \exp\left(\frac{k^2}{n^2 \sigma^2}\right) f(\theta^{**}, \vec{Z}),
\end{eqnarray}
with the middle step using \eqref{eq:theta''_subg} and the other two steps using the definition of $f$ in \eqref{eq:f_subg}.  Taking the product over all blocks $j=1,\dotsc,n/k-1$, we obtain
\begin{equation}
    \prod_j f(\theta^*, \vec{Z}^{j}) \leq \prod_j\exp\left(\frac{k^2}{n^2 \sigma^2}\right) f(\theta^{**}, \vec{Z}^{j}) \leq 2 \prod_j f(\theta^{**}, \vec{Z}^{j}),
\label{eq:rounding}
\end{equation}
where the last step holds when $n/k$ is sufficiently large (we have assumed that it approaches $\infty$ as $n \to \infty$).

Let $T_{\epsilon} = T \setminus [\theta^*-2\epsilon+\frac2{n^2}, \theta^*+2\epsilon-\frac2{n^2}]$.  If there exists some $\theta' \in T_{\epsilon}$ such that $\prod_j f(\theta', \vec{Z}^j) \geq \prod_j f(\theta^*,\vec{Z}^j)$, then \eqref{eq:rounding} gives $2 \prod_j f(\theta', \vec{Z}^j) \geq 2\prod_j f(\theta^{**}, \vec{Z}^j)\geq \prod_j f(\theta^{*},\vec{Z}^j)$.  Moreover, since $|\theta^{**}-\theta^*|\leq \frac1{n^2}$ and $|\theta^{**}-\theta'| \ge 2\epsilon - \frac{2}{n^2}$, the triangle inequality gives
\begin{eqnarray}
    |\theta'-\theta^*| \geq 2\epsilon - \frac{3}{n^2}. \label{eq:rounding_subg}
\end{eqnarray}
We can now follow similar steps to \eqref{eq:setT_step1}--\eqref{eq:setT_step4} as follows:
\begin{eqnarray}
    & &\p \Bigg(\bigcup_{\theta' \in T_{\epsilon}} \bigg\{\prod_j f(\theta',\vec{Z}^j) \ge \prod_j f(\theta^{**},\vec{Z}^j) \bigg\}\Bigg) \\
    &\leq& \p \Bigg(\bigcup_{\theta' \in T_{\epsilon}} \bigg\{ 2\prod_j f(\theta'. \vec{Z}^j) \ge \prod_j f(\theta^{*},\vec{Z}^j) \bigg\}\Bigg) \label{eq:T_subg_step1}\\
    & \leq& \sum_{\theta' \in T_{\epsilon}} \p \Bigg( 2\prod_j f(\theta', \vec{Z}^j) \ge \prod_j f(\theta^*,\vec{Z}^j)\Bigg) \label{eq:T_subg_step2}\\
    & \leq & 2\sum_{\theta' \in T_{\epsilon}} \mathbb{E} \Bigg( \prod_j \frac{f(\theta',\vec{Z}^j)}{f(\theta^*, \vec{Z}^j)}\Bigg) \label{eq:T_subg_step3}\\
    &\leq & \exp\bigg( -(n-o(n))\min\Big(\frac{(\theta^*-\theta')^2}{8\sigma^2}, E^{\rm chan}(0)\Big)\bigg) \label{eq:T_subg_step4}\\
    &\leq & \exp\bigg( -(n-o(n))\min\Big(\frac{\epsilon^2}{2\sigma^2}+O(1/n^2), E^{\rm chan}(0)\Big)\bigg) \label{eq:T_subg_step5}\\
    &\leq & \exp\bigg( -(n-o(n))\min\Big(\frac{\epsilon^2}{2\sigma^2}, E^{\rm chan}(0)\Big)\bigg), \label{eq:T_subg_step6}
\end{eqnarray}
where \eqref{eq:T_subg_step1} was established in the previous paragraph, \eqref{eq:T_subg_step2} applies the union bound, \eqref{eq:T_subg_step3} applies Markov's inequality, \eqref{eq:T_subg_step4} applies Lemma \ref{lem:Ef_subg}, and \eqref{eq:T_subg_step5} applies \eqref{eq:rounding_subg} and the fact that $|T_{\epsilon}|$ has polynomial size. 

Equation \eqref{eq:T_subg_step6} upper bounds the probability that $\theta^{**} \notin S$, and recalling that this in turn upper bounds the error probability, it follows that the error exponent $\min\big(\frac{\epsilon^2}{2\sigma^2}, E^{\rm chan}(0)\big)$ is achievable, as desired.

\subsection{Proof of Corollary \ref{cor:subg_conv} (Converse for Gaussian sources)} \label{sec:pf_subg_conv}

We apply the general converse result for non-causal protocols from Lemma \ref{lem:noncausal}, and then rewrite the channel term using $c_M$ from Definition \ref{def:cM}.  It then only remains to show that $E^{\rm src}_{\epsilon} = \frac{\epsilon^2}{2\sigma^2}$.  It follows immediately from Theorem \ref{thm:subg_ach} (e.g., paired with a noiseless binary channel) that $E^{\rm src}_{\epsilon} \ge \frac{\epsilon^2}{2\sigma^2}$, so it suffices to show a matching converse, i.e., $E^{\rm src}_{\epsilon} \le \frac{\epsilon^2}{2\sigma^2}$.

To see this, we simply observe that attaining accuracy $\epsilon$ implies being able to successfully perform a hypothesis test between two Gaussians with means separated by $2\epsilon'$, for any $\epsilon' > \epsilon$.  It is well-known that such a hypothesis test has an optimal error exponent of $\frac{(\epsilon')^2}{2\sigma^2}$ (e.g., see \cite[Sec.~11.9]{cover_thomas}).  Taking $\epsilon' \to \epsilon$ from above gives the desired result.

\section{Further Extensions and Discussion} \label{sec:discussion}

\subsection{Heavy-tailed random variables} \label{sec:heavy}

In our analysis, the only place where we used the sub-Gaussian style concentration property is in \eqref{eq:Ef2_step2}.  This observation lets us readily handle more general distributions by \emph{letting $\alpha$ itself be a more general ``base'' estimator for a single block} (not necessarily the sample mean).

In more detail, suppose that there exists an estimator $\alpha(\vec{X})$ operating on blocks $\vec{X}$ of size $k$ and satisfying the following properties:
\begin{enumerate}
    \item[(i)] The estimator satisfies the following for some $\sigma^2 > 0$ and any $\epsilon \in \big(0,\frac12\big)$:
    \begin{equation}
        \mathbb{P}(|\alpha(\vec{X}) - \theta^*| \geq \epsilon) \leq 2 \exp\Big(-\frac{k\epsilon^2 }{2\sigma^2} (1+o(1))\Big). \label{eq:heavy_bound}
    \end{equation}
    \item[(ii)] In the case that $X$ is a multiple of $1/k$ with probability one, there are only at most ${\rm poly}(k)$ values within the range $[-k,k]$ that the estimator can output.
    \item[(iii)] Property (i) continues to hold (possibly with a modified $1+o(1)$ term) after rounding to a multiple of $1/k$ in the manner described following Assumption \ref{as:rounding}.\footnote{With only minor adjustments, we can generalize $1/k$ to $1/k^C$ for any $C > 0$.}
\end{enumerate}
Under these conditions, we can apply the exact same analysis as the sub-Gaussian case, with the encoder now forming $\alpha$ using the base estimator rather than the sample mean.  The rest of the analysis is unchanged, and the error exponent in \eqref{eq:subg_ach} is again achievable.

As a well-known example, we note that the \emph{median-of-means} estimator with suitably-chosen parameters satisfies property (i) above with $\sigma^2 = 32 {\rm Var}[X]$ \cite[Thm.~2]{lugosi2019mean}, and properties (ii) and (iii) are straightforward to verify.\footnote{For property (ii), note that each of the ($k$ or fewer) intermediate means in the median-of-means method takes on one of ${\rm poly}(k)$ values (similar to Section \ref{sec:subg_protocol}), and the median will equal one of these intermediate means.   
For property (iii), similar to Section \ref{sec:subg_protocol}, we can write the rounded distribution as $X+\Delta$ and then use ${\rm Var}[X+\Delta] = {\rm Var}[X]+2{\rm Cov}[X,\Delta] + {\rm Var}[\Delta]$, which simplifies to  ${\rm Var}[X]+O(1/k)$ by writing ${\rm Cov}[X,\Delta] \le \sqrt{{\rm Var}[X]{\rm Var}[\Delta]}$ and ${\rm Var}[\Delta] = O(1/k^2)$.}  
Thus, this estimator attains sub-Gaussian like concentration assuming nothing other than finite variance.  The above argument allows us to transfer this guarantee for direct observations to an error exponent for our setting with coded relayed observations.

\subsection{Estimation beyond $\theta^* \in [0,1]$} \label{sec:theta_range}

In Section \ref{sec:subgaussian} we mentioned that for sub-Gaussian sources some restriction is needed on $\theta^*$, and we focused on $\theta^* \in [0,1]$.  However, the protocol and analysis easily extends to $\theta^* \in [-C,C]$ for any $C = {\rm poly}(k)$.  Since our analysis permits setting $k = \sqrt{n}$, this means that we can handle any $C = {\rm poly}(n)$.  The idea is that in the definition of $\alpha$ in \eqref{eq:alpha_def}, we clip the rewards to $[-C-k,C+k]$ instead of $[-k,k]$, and we are still left with only polynomially many possible outcomes.  The rest of the analysis is essentially unchanged except for the polynomial pre-factors. 

We also note that our converse results are based on Lemma \ref{lem:noncausal}, and if we generalize from $\theta^* \in [0,1]$ to $\theta^* \in [-C,C]$ then we should replace $E^{\rm chan}_{\lceil 1/(2\epsilon) \rceil}$ by $E^{\rm chan}_{\lceil C/\epsilon \rceil}$ therein.

\subsection{Vector-valued sources} \label{sec:vector}

Consider the case where the source is $d$-dimensional, and the setup from one of our previous sections (Bernoulli in Section \ref{sec:bernoulli} or sub-Gaussian in Section \ref{sec:subgaussian}) applies in every entry indexed by $i \in 1,\dotsc,d$.  We claim that in such a scenario, we can achieve the same error exponent in terms of $\epsilon$-accuracy in each component (separately) as what we achieved for the case that $d=1$.

To see this, we simply modify the protocol such that if $f(k)$ codewords were used when $d=1$, then $f(k)^d$ codewords are used more generally, and these are used to encode the behavior of each of the $d$ entries (e.g., the 
total number of 1s in each entry in the Bernoulli case).  Since $f(k)$ is polynomial, so is $f(k)^d$ for any constant $d$, and thus we can still use Lemma \ref{lem:bsc_hamming_distance} and Corollary \ref{cor:db_dmc}.  The subsequent analysis is essentially unchanged except that there are more $\alpha$ values, but still only polynomially many.  Moreover, since the error events have exponentially small probability, we can safely take a union bound over the $d$ components.

However, it may also be of interest to understand error guarantees beyond the entry-wise one, e.g., the squared $\ell_2$-error.  In such scenarios, entry-wise analyses appear to be insufficient, and parts of our protocol may require non-trivial changes (e.g., how to generalize the choice $\hat{\theta}_n = \frac{1}{2}\big( \min S + \max S \big)$).  Such considerations are left for future work.

% We can generalize our protocols to the situation where the source is a $d$-dimensional vector instead. The teacher can simply create $(4k^3)^d$ randomly generated codewords and the analysis still holds component wise.  {\bf \color{magenta} [TODO: Regarding vector valued sources, that is right. We just imagine that there are $(k^3)^{d-1}$ codewords for each alpha, and this blows up the polynomial prefactor. We handle things the same way, noting that for the functions $f,g$ we need to sum up over the larger set of alpha. Otherwise nothing changes.]} If we define the error event by the union of error events, then the same error exponent holds. However, for error guarantees defined with respect to other norms (e.g., $\ell_1$ or $\ell_2$), it appears to be insufficient to study the components separately; we leave such considerations for future work.

\subsection{Discussion on computational complexity} \label{sec:computation}

While our focus is on information-theoretic limits rather than computational considerations, we note that our protocol has polynomial time, even when an ``unstructured'' (e.g., random) code is used for the codebook.  At the teacher all that needs to be done is to compute $\alpha$ and transmit the corresponding codeword.  At the student, there are more steps of computation involved; starting with the Bernoulli+BSC case, we note the following:
\begin{itemize}
    \item Regarding the codebook used, the $k+1$ codewords of length $k$ can be stored in a size-$O(k^2)$ lookup table.
    \item The function $g$ in \eqref{eq:g_def} can be computed in $O(k)$ time by using a product over $k$ symbols to evaluate $\p(\vec{Z}| \vec{W}_\alpha)$.
    \item The function $f$ in \eqref{eq:f_subg} can be computed in $O(k^2)$ time by summing over $k+1$ values of $\alpha$, each of which computes $g$.
    \item We can compute any given $\prod_j f(\theta, \vec{Z}^j)$ in \eqref{eq:setS} in $O(nk)$ time since there are $n/k-1$ values of $j$.
    \item Hence, we can compute the set $S$ (and thus $\hat{\theta}_n$) in time $O(n^2 k)$ since there are $O(n)$ values of $\theta$ to consider.
\end{itemize}
In the sub-Gaussian case, a similar argument applies, but there are $O(k^3)$ values of $\alpha$ and $O(n^2$) values of $\theta$, so the computation time is $O(n^3 k^3)$.  This is perhaps higher than ideal, but nevertheless polynomial time.  For both the Bernoulli+BSC case and the sub-Gaussian case, an interesting direction for future work would be to attain similar results with a more efficient decoder (e.g., $n^{1+o(1)}$-time).

\section{Conclusion} \label{sec:conclusion}

We have studied large deviations bounds for the problem of statistical mean estimation with coded relayed observations.  For both Bernoulli and sub-Gaussian sources, and both BSCs and general DMCs, we have provided a block-structured protocol whose error exponent is optimal in broad cases of interest.

We conclude by highlighting some directions that may be of interest in future work:
\begin{enumerate}
    \item Our current results are strongest in medium-accuracy and high-accuracy regimes (i.e., small-to-moderate $\epsilon$).  In contrast, gaps still remain in lower-accuracy regimes (i.e., high $\epsilon$), and in such cases it remains unclear whether or not there exists a gap between the causal and non-causal settings.
    \item We only provided partial results for vector-valued sources, considering component-wise recovery.  It remains open to handle other performance measures such as $\ell_2^2$ error, which may require different techniques.
    \item We focused on minimax bounds (over $\theta^* \in [0,1]$ and/or over the sub-Gaussian class), and it would be of interest to better understand refined instance-dependent results (e.g., depending on $\theta^*$ in the Bernoulli case or depending on other distributional properties in the sub-Gaussian case).
    \item Our protocols rely on the teacher knowing both $\epsilon$ and $P_{Z|W}$, as well as $\sigma^2$ in the sub-Gaussian setting, whereas ideally one might use a  ``universal'' strategy that does not require such knowledge.
    \item While our protocols are polynomial-time, the polynomial powers are higher than ideal, and generally speaking we do not claim our protocols to be ``practical''.  It is therefore a natural direction to study protocols that have a smaller runtime and/or are specifically targeted at attaining strong experimental performance.
\end{enumerate}

\appendix

\subsection{Proof of Lemma \ref{lem:noncausal} (Non-causal setting)} \label{sec:noncausal}

\subsubsection{Converse part}

The first term in the converse is trivial by the definition of $E^{\rm src}_{\epsilon}$.  For the second term, we consider the case that $\theta^*$ is restricted to be an integer multiple of $\epsilon'$ for some $\epsilon' > \epsilon$.  Since $\theta^* \in [0,1]$, by letting $\epsilon'$ be sufficiently close to $\epsilon$, the number of such $\theta^*$ values is $M = \big\lceil \frac{1}{2\epsilon} \big\rceil$.  For example, if $\epsilon = 1/8$ then the points are approximately $\{0,1/4,1/2,3/4\}$ but with the gaps expanded very slightly, giving $M = 4$, but for slightly smaller $\epsilon$ this would increase to $M=5$ due to adding a fifth point near 1.

Observe that even if the teacher knows $\theta^*$ exactly, it is still required to send the student sufficient information to identify the value from the size-$M$ set.  This amounts to sending one of $M = \big\lceil 1/(2\epsilon)\big\rceil$ messages over the channel, and the associated exponent is $E^{\rm chan}_M = E^{\rm chan}_{\lceil 1/(2\epsilon)\rceil}$, thus yielding the desired second term.

\subsubsection{Achievability part}

For the achievability part, we fix $\delta \in (0,1)$ and consider the following protocol for the non-causal setting:
\begin{itemize}
    \item The teacher uses an optimal estimation procedure for direct observations to estimate $\theta^*$ to within accuracy $(1-\delta)\epsilon$; this estimate is denoted by $\tilde{\theta}$.
    \item The teacher rounds $\tilde{\theta}$ to the nearest point in a set $\calI$ of points in $[0,1]$ . 
    We design this set such that every $\theta \in [0,1]$ has distance at most $\delta \epsilon$ to the nearest point.  For instance, this can be achieved as follows:
    \begin{itemize}
        \item Start with the set $\calI = \{\delta\epsilon, 3\delta\epsilon, 5\delta\epsilon,\dotsc \} \cap [0,1]$;
        \item If $\theta = 1$ is not $\delta\epsilon$-close to the last point, then add $\theta=1$ itself to $\calI$.
    \end{itemize}
    With this construction, the size of $\calI$ is $\big\lceil\frac{1}{2\delta\epsilon}\rceil$; for example, if $\delta\epsilon = \frac{1}{4}$ then $\calI = \{1/4, 3/4\}$, but if $\epsilon$ is decreased slightly then we need to add a third point (namely, 1).  After rounding the initial estimate to $\calI$ (i.e., performing quantization), the corresponding index is sent to the student using an error-correcting code of length $n$ with $|\calI|$ messages.
    \item The student decodes the received sequence to obtain the corresponding point in $\calI$, and this forms the final estimate $\hat{\theta}_n$.
\end{itemize}
The teacher's own estimation incurs an error of up to $(1-\delta)\epsilon$, and the subsequent rounding incurs up to another $\delta\epsilon$ for a total of at most $\epsilon$.  Since the number of codewords in the codebook is $\big\lceil \frac{1}{2\delta\epsilon} \big\rceil$, the result readily follows.

\begin{remark}
    In the above analysis, we consider uniform quantization for simplicity, but improved quantization strategies may be possible, e.g., quantizing more finely in regions where estimation is known to be more difficult.  (In particular, see \eqref{eq:e_theta} for a relevant $\theta^*$-dependent bound in the Bernoulli case.)  The above strategy is only introduced as a baseline, rather than something that we seek to optimize carefully.
\end{remark}

\subsection{Proof of Lemma \ref{lem:simple_forward} (Achievability via simple forwarding)} \label{sec:simple}

Observe that under simple forwarding, the student has access to $n-1$ i.i.d.~observations drawn from ${\rm Ber}( \tilde{\theta}^* )$, where
\begin{eqnarray}
    \tilde{\theta}^* = \theta^* (1-p) + (1-\theta^*)p = \theta^* (1-2p) + p.
\end{eqnarray}
Accordingly, the student may apply any estimation procedure to form an estimate $\tilde{\theta}$ of $\tilde{\theta}^*$, and then form the final estimate by shifting and scaling:
\begin{equation}
    \hat{\theta} = \frac{\tilde{\theta} - p}{1-2p}.
\end{equation}
Clearly, if the estimate $\hat{\theta}$ is $\epsilon$-accurate if and only if the estimate $\tilde{\theta}$ is $\epsilon(1-2p)$-accurate, and the desired result then follows from the definition of $E^{\rm src}_{\epsilon}$ (the distinction between $n-1$ vs.~$n$ samples does not impact the error exponent).

\subsection{Proof of Lemma \ref{lem:est_forward} (Achievability via one-shot estimate-and-forward)} \label{sec:oneshot}

The analysis is identical to the achievability part in Appendix \ref{sec:noncausal}, except that the teacher's block length is reduced from $n$ to $(1-\lambda)n$, and the student's block length is reduced from $n$ to $\lambda n$.  The details are omitted to avoid repetition.

\subsection{Proof of Corollary \ref{cor:db_dmc} (Existence of good codebooks for general DMCs)} \label{sec:pf_db_dmc}

    Let $P_W$ be the distribution achieving the maximum in \eqref{eq:e0_q}. Consider random coding in which each symbol of each codeword is randomly drawn according to $P_W$.  Then, we can write $ \db(\vec{W},\vec{W}')$ as a sum of i.i.d.~variables: 
    \begin{equation}
        \db(\vec{W},\vec{W}',P_{Z|W}) = \sum_{i} \db(\vec{W}_i, \vec{W}'_i,P_{Z|W}), \label{eq:db_W_decomp}
    \end{equation}
    where the equality follows from the tensorization of $\db$ (Lemma \ref{lem:tens_db}). 
    The expected value is $k\cdot E^{\rm chan}(0)$, and we proceed to separately study the cases $E^{\rm chan}(0) < \infty$ and $E^{\rm chan}(0) = \infty$.
    
    When $E^{\rm chan}(0) < \infty$, the terms $\db(\vec{W}_i, \vec{W}'_i,P_{Z|W})$ in \eqref{eq:db_W_decomp} are all finite.  As a result, for any $c<1$, the probability that $\db(\vec{W}, \vec{W}',P_{Z|W}) < ck E^{\rm chan}(0)$ is exponentially small in $c$ by Hoeffding's inequality \cite[Sec.~2.6]{boucheron}.  Since there are polynomially many codewords, by the union bound, it holds with high probability that $\db(\vec{W},\vec{W}',P_{Z|W}) \geq (k-o(k)) E^{\rm chan}(0)$.

    In the case that $E^{\rm chan}(0) = \infty$, there must exist two symbols $w,w'$ with $\db(w,w',P_{Z|W}) = \infty$ and $P_W(w)P_W(w')>0$.  Moreover, we have $\db(\vec{W},\vec{W}') = \infty$ as long as there exists a single position $i \in \{1,\dotsc,n\}$ where one codeword has $w$ and the other has $w'$.  Since the codewords are i.i.d., this occurs with probability approaching 1 exponentially fast, and thus a similar argument to the case $E^{\rm chan}(0) < \infty$ applies.

\subsection{Proof of Lemma \ref{lem:simplification_Bern} (Simplification of exponent for the Bernoulli/BSC setting)} \label{app:pf_simplification_Bern}

We seek to show that $|\theta'-\theta^*|\geq 2\epsilon$ implies $\db(\theta^*, \theta') \geq D(\frac12 \| \frac12 + \epsilon)$.  Using the shorthand $x = \frac{\theta^*+\theta'}{2}$ and $y = \frac{\theta^*-\theta'}{2}$, we have
\begin{equation}
    \rho(\theta^*, \theta') = \sqrt{\theta^*\theta'} + \sqrt{(1-\theta^*)(1-\theta')} = \sqrt{(x+y)(x-y)} + \sqrt{(1-x-y)(1-x+y)}.
    \label{eq:rho_xy_expand}
\end{equation}
We will now show that \eqref{eq:rho_xy_expand} achieves its maximum at $x=\frac12$. Differentiating with respect to $x$, we have\footnote{It is straightforward to check from the definitions of $x$ and $y$ that the terms inside the square root in \eqref{eq:first_der} are never negative.}
\begin{equation}
    \frac{\partial}{\partial x} \Big( \sqrt{(x+y)(x-y)} + \sqrt{(1-x-y)(1-x+y)} \Big) = \frac{x}{\sqrt{x^2-y^2}} - \frac{(1-x)}{\sqrt{(1-x)^2 - y^2}}. \label{eq:first_der}
\end{equation}
It follows that when $x=\frac12$, the derivative is 0.

We now compute the second derivative:
\begin{equation}
    \frac{\partial^2}{\partial x^2} \Big(\sqrt{(x+y)(x-y)} + \sqrt{(1-x-y)(1-x+y)}\Big) = -\frac{y^2}{(x^2-y^2)^{3/2}} - \frac{y^2}{((1-x)^2-y^2)^{3/2}} < 0,\ \forall x,y \label{eq:second_der}
\end{equation}
so $x=\frac12$ is a global maximum. Therefore,
\begin{equation}
    \rho(\theta^*, \theta') \leq 2\sqrt{\Big(\frac12 + y\Big)\Big(\frac12 -y\Big)} \leq 2\sqrt{\frac{1}{4} - \epsilon^2} = \exp\bigg(-D\Big(\frac12 \,\Big\|\, \frac12 + \epsilon\Big)\bigg),
\end{equation}
where the middle step uses $\big(\frac12 + y\big)\big(\frac12 -y\big) = \frac{1}{4}-y^2$ along with $|y| \ge \epsilon$ (since $y = \frac{\theta^*-\theta'}{2}$ and $|\theta'-\theta^*|\geq 2\epsilon$), and the last step uses Lemma \ref{lem:eps_equiv}.  Taking the negative logarithm, we obtain
\begin{equation}
    \db(\theta^*, \theta') \geq D\Big(\frac12 \Big\| \frac12 + \epsilon\Big),
\end{equation}
as desired. 

% and substituting this into \eqref{eq:e_theta} gives
% \begin{equation}
%     E \geq \min\bigg( D\Big(\frac12 \big\| \frac12 + \epsilon\Big), \frac12 D\Big(\frac12 \big\| p\Big)\bigg).
% \end{equation}

\bibliographystyle{IEEEtran}
\bibliography{general}
\end{document}